\documentclass[10pt,twocolumn,twoside]{article}
\usepackage[english]{babel} 
\usepackage[cm]{fullpage}

\usepackage{amssymb}    
\usepackage{float}     
\usepackage{epic}       
\usepackage{fancybox}   
\usepackage{graphicx}
\usepackage{amsmath}
\usepackage{amsthm} 
\usepackage{tikz}
\usepackage{pgfplots}
\usepackage{amsfonts}
\usepackage{hyperref} 
\usepackage[ruled,vlined]{algorithm2e}

\usetikzlibrary{calc,trees,fadings,fit,positioning,arrows,shapes.symbols,decorations.markings,shapes.geometric}

\usetikzlibrary{external}
\ifnum\pdfshellescape=1
	\tikzsetexternalprefix{pgffigures/}
	\tikzsetfigurename{unnamed}
\else
\fi

\newcommand{\RTTCYCLES}{
\begin{figure}[ht]
	\begin{center}
	\tikzsetnextfilename{RTT_cycles}
	\def\xw{20pt}
	\begin{tikzpicture}[
	vertex/.style={circle,draw,solid,fill=blue!40,font=\tiny,inner sep=0pt,minimum width=1.2em},x=\xw,y=\xw,
	normal/.style={densely dotted},
	hcycle/.style={green,solid,very thick},vcycle/.style={blue,solid,very thick}]
	\begin{scope}[red,dashed,thick] 
	\draw (0,0) -- (8,0) coordinate (endaxis1);
	\draw (0,0) -- (0,4) coordinate (endaxis2);
	\draw[solid,|->] (endaxis1) -- node[auto]{$\mathbf e_1$} ++(1,0);
	\draw[solid,|->] (endaxis2) -- node[auto]{$\mathbf e_2$} ++(0,1);
	\end{scope}
	\foreach \i in {0,1,2,3,4,5,6,7}
	\foreach \j in {0,1,2,3}
	{
		\node[vertex](nodo\i p\j) at (\i,\j) {};
	}
	\foreach \i in {0,1,2,3,4,5,6}
	\foreach \j in {0,1,2,3}
	{
		\pgfmathparse{int(\i+1)}
		\let\r=\pgfmathresult
		\pgfmathtruncatemacro\cond{\j==0}
		\ifthenelse{\cond=1}{
			\draw[hcycle] (nodo\i p\j) -- (nodo\r p\j);
		}{
			\draw[normal] (nodo\i p\j) -- (nodo\r p\j);
		}
	}
	\foreach \i in {0,1,2,3,4,5,6,7}
	\foreach \j in {0,1,2}
	{
		\pgfmathparse{int(\j+1)}
		\let\r=\pgfmathresult
		\pgfmathtruncatemacro\cond{(\i==0)||(\i==4)}
		\ifthenelse{\cond=1}{
			\draw[vcycle] (nodo\i p\j) -- (nodo\i p\r);
		}{
			\draw[normal] (nodo\i p\j) -- (nodo\i p\r);
		}
	}
	\foreach \j in {0,1,2,3}
	{
		\pgfmathtruncatemacro\cond{\j==0}
		\ifthenelse{\cond=1}{
			\draw[hcycle] (nodo7p\j) edge[overlay,out=-10,in=-170] (nodo0p\j);
		}{
			\draw[normal] (nodo7p\j) edge[overlay,out=-10,in=-170] (nodo0p\j);
		}
	}
	\foreach \i in {0,1,2,3,4,5,6,7}
	{
		\pgfmathparse{int(mod(\i+4,8))}
		\let\r=\pgfmathresult
		\pgfmathtruncatemacro\cond{(\i==0)||(\i==4)}
		\ifthenelse{\cond=1}{
			\draw[vcycle] (nodo\i p3) edge[overlay,out=90,in=-90] (nodo\r p0);
		}{
			\draw[normal] (nodo\i p3) edge[overlay,out=90,in=-90] (nodo\r p0);
		}
	}
	\draw (-1,-1);
	\draw (8,4);
	\end{tikzpicture}
	\end{center}
	\caption{Two perpendicular cycles of length 8 in the $RTT(4)$.}
	\label{fig:RTTcycles}
\end{figure}
}

\newcommand{\JOININGCYCLE}{
\begin{figure}[h]
	\begin{center}
	\tikzsetnextfilename{joiningcycle}
	\begin{tikzpicture}[scale=0.4,
		nodo/.style={draw,circle},
		]
	\foreach \k in {0,1,2,3}
	\foreach \i in {0,1,2,3}
	\foreach \j in {0,1,2,3}
	{
		\pgfmathtruncatemacro\cond{(\i==0)&&(\j==0 || \j==2)}
		\ifthenelse{\cond=1}{
			\def\filling{blue!40}
		}{
			\def\filling{none}
		}
		\path (5*\k+\i,\j) node[nodo,fill=\filling] (nodo_\i_\j_\k) {};
	}
	\begin{scope}[wrap/.style={densely dotted}]
		\foreach \k in {0,1,2,3}
		\foreach \i in {0,1,2}
		\foreach \j in {0,1,2,3}
		{
			\pgfmathtruncatemacro\next{\i+1}
			\draw (nodo_\i_\j_\k) -- (nodo_\next_\j_\k);
		}
		\foreach \k in {0,1,2,3}
		\foreach \j in {0,1,2,3}
			\draw[wrap] (nodo_3_\j_\k) to[out=-10,in=190] (nodo_0_\j_\k);
		\foreach \k in {0,1,2,3}
		\foreach \i in {0,1,2,3}
		\foreach \j in {0,1,2}
		{
			\pgfmathtruncatemacro\next{\j+1}
			\draw (nodo_\i_\j_\k) -- (nodo_\i_\next_\k);
		}
		\foreach \k in {0,1,2,3}
		\foreach \i in {0,1,2,3}
			\draw[wrap] (nodo_\i_3_\k) to[out=80,in=-80] (nodo_\i_0_\k);
	\end{scope}
	\begin{scope}[very thick,overlay,->]
	\draw[blue!100!green] (nodo_0_0_0) -- (nodo_0_0_1);
	\draw[blue!86!green] (nodo_0_0_1) -- (nodo_0_0_2);
	\draw[blue!71!green] (nodo_0_0_2) -- (nodo_0_0_3);
	\draw[blue!57!green] (nodo_0_0_3) to[out=0,in=180] (nodo_0_2_0);
	\draw[blue!43!green] (nodo_0_2_0) -- (nodo_0_2_1);
	\draw[blue!29!green] (nodo_0_2_1) -- (nodo_0_2_2);
	\draw[blue!14!green] (nodo_0_2_2) -- (nodo_0_2_3);
	\draw[blue!00!green] (nodo_0_2_3) to[out=0,in=180] (nodo_0_0_0);
	\end{scope}
	\draw (-1.5,0);
	\end{tikzpicture}
	\end{center}
	\caption{$\langle \mathbf e_n\rangle$ going by the disjoint copies of $\mathcal G(B)$.}
	\label{fig:joiningcycle}
\end{figure}}

\newcommand{\JOININGCYCLETRIDIMENSIONAL}{
\begin{figure}[h]
	\begin{center}
	\tikzsetnextfilename{joiningcycletridimensional}
	\begin{tikzpicture}[scale=0.4,
		nodo/.style={draw,circle},
		x={(1.8,0)},
		y={(0,1.8)},
		z={(2.1,-.4)},
		]
	\foreach \k in {0,1,2,3}
	{
	\foreach \i in {0,1,2,3}
	\foreach \j in {0,1,2,3}
	{
		\pgfmathtruncatemacro\cond{(\i==0)&&(\j==0 || \j==2)}
		\ifthenelse{\cond=1}{
			\def\filling{blue!40}
		}{
			\def\filling{white}
		}
		\path (\i,\j,\k) node[nodo,fill=\filling] (nodo_\i_\j_\k) {};
	}
	\begin{scope}[wrap/.style={densely dotted,overlay}]
		\foreach \i in {0,1,2}
		\foreach \j in {0,1,2,3}
		{
			\pgfmathtruncatemacro\next{\i+1}
			\draw (nodo_\i_\j_\k) -- (nodo_\next_\j_\k);
		}
		\foreach \j in {0,1,2,3}
			\draw[wrap] (nodo_3_\j_\k) to[out=-10,in=190] (nodo_0_\j_\k);
		\foreach \i in {0,1,2,3}
		\foreach \j in {0,1,2}
		{
			\pgfmathtruncatemacro\next{\j+1}
			\draw (nodo_\i_\j_\k) -- (nodo_\i_\next_\k);
		}
		\foreach \i in {0,1,2,3}
			\draw[wrap] (nodo_\i_3_\k) to[out=80,in=-80] (nodo_\i_0_\k);
	\end{scope}
	}
	\begin{scope}[very thick,overlay,->]
	\draw[blue!100!green] (nodo_0_0_0) -- (nodo_0_0_1);
	\draw[blue!86!green] (nodo_0_0_1) -- (nodo_0_0_2);
	\draw[blue!71!green] (nodo_0_0_2) -- (nodo_0_0_3);
	\draw[blue!57!green] (nodo_0_0_3) to[out=-10,in=170] (nodo_0_2_0);
	\draw[blue!43!green] (nodo_0_2_0) -- (nodo_0_2_1);
	\draw[blue!29!green] (nodo_0_2_1) -- (nodo_0_2_2);
	\draw[blue!14!green] (nodo_0_2_2) -- (nodo_0_2_3);
	\draw[blue!00!green] (nodo_0_2_3) to[out=-10,in=170] (nodo_0_0_0);
	\end{scope}
	\draw (-1,3.5);
	\draw (10,-2.);
	\end{tikzpicture}
	\end{center}
	\caption{$\langle \mathbf e_n\rangle$ going by the disjoint copies of $\mathcal G(B)$.}
	\label{fig:joiningcycle}
\end{figure}}

\newcommand{\JOININGCYCLETRIDIMENSIONALSHADOW}{
\begin{figure}[h]
	\begin{center}
	\tikzsetnextfilename{joiningcycletridimensionalshadow}
	\begin{tikzpicture}[scale=0.4,
		nodo/.style={draw,circle},
		x={(1.8,0)},
		y={(0,1.8)},
		z={(2.1,-.4)},
		]
	\foreach \k/\op in {0/1,1/.3,2/.3,3/1}
	{
	\begin{scope}[opacity=\op]
		\foreach \i in {0,1,2,3}
		\foreach \j in {0,1,2,3}
		{
			\pgfmathtruncatemacro\cond{(\i==0)&&(\j==0 || \j==2)}
			\ifthenelse{\cond=1}{
				\def\filling{blue!40}
				\def\opnode{2.0*\op}
			}{
				\def\filling{white}
				\def\opnode{\op}
			}
			\path (\i,\j,\k) node[nodo,fill=\filling,opacity=\opnode] (nodo_\i_\j_\k) {};
		}
		\begin{scope}[wrap/.style={densely dotted,overlay}]
			\foreach \i in {0,1,2}
			\foreach \j in {0,1,2,3}
			{
				\pgfmathtruncatemacro\next{\i+1}
				\draw (nodo_\i_\j_\k) -- (nodo_\next_\j_\k);
			}
			\foreach \j in {0,1,2,3}
				\draw[wrap] (nodo_3_\j_\k) to[out=-10,in=190] (nodo_0_\j_\k);
			\foreach \i in {0,1,2,3}
			\foreach \j in {0,1,2}
			{
				\pgfmathtruncatemacro\next{\j+1}
				\draw (nodo_\i_\j_\k) -- (nodo_\i_\next_\k);
			}
			\foreach \i in {0,1,2,3}
				\draw[wrap] (nodo_\i_3_\k) to[out=80,in=-80] (nodo_\i_0_\k);
		\end{scope}
	\end{scope}
	}
	\begin{scope}[very thick,overlay,->]
	\draw[blue!100!green] (nodo_0_0_0) -- (nodo_0_0_1);
	\draw[blue!86!green] (nodo_0_0_1) -- (nodo_0_0_2);
	\draw[blue!71!green] (nodo_0_0_2) -- (nodo_0_0_3);
	\draw[blue!57!green] (nodo_0_0_3) to[out=-10,in=170] (nodo_0_2_0);
	\draw[blue!43!green] (nodo_0_2_0) -- (nodo_0_2_1);
	\draw[blue!29!green] (nodo_0_2_1) -- (nodo_0_2_2);
	\draw[blue!14!green] (nodo_0_2_2) -- (nodo_0_2_3);
	\draw[blue!00!green] (nodo_0_2_3) to[out=-10,in=170] (nodo_0_0_0);
	\end{scope}
	\draw (-1,3.5);
	\draw (10,-2.);
	\end{tikzpicture}
	\end{center}
	\caption{$\langle \mathbf e_n\rangle$ going by the disjoint copies of $\mathcal G(B)$.}
	\label{fig:joiningcycle}
\end{figure}}

\newcommand{\JOININGCYCLETRIDIMENSIONALSHADOWAXIS}{
\begin{figure}[h]
	\begin{center}
	\tikzsetnextfilename{joiningcycletridimensionalshadowaxis}
	\begin{tikzpicture}[scale=0.4,
		nodo/.style={draw,circle},
		x={(1.8,0)},
		y={(0,1.8)},
		z={(1.2,-.4)},
		]
	\begin{scope}[red,dashed,thick] 
	\draw (0,0) -- (7,0) coordinate (endaxis1);
	\draw (0,0) -- (0,4) coordinate (endaxis2);
	\draw (0,0) -- (0,0,4.5) coordinate (endaxis3);
	\draw[solid,|->] (endaxis1) -- node[auto]{$\mathbf e_1$} ++(1,0);
	\draw[solid,|->] (endaxis2) -- node[auto]{$\mathbf e_2$} ++(0,1);
	\draw[solid,|->] (endaxis3) -- node[auto]{$\mathbf e_3$} ++(0,0,1);
	\end{scope}
	\foreach \k/\op in {0/1,1/.4,2/.4,3/1}
	{
	\begin{scope}[opacity=\op]
		\foreach \i in {0,1,2,3}
		\foreach \j in {0,1,2,3}
		{
			\pgfmathtruncatemacro\cond{(\i==0)&&(\j==0 || \j==2)}
			\ifthenelse{\cond=1}{
				\def\filling{blue!40}
				\def\opnode{2.0*\op}
			}{
				\def\filling{white}
				\def\opnode{\op}
			}
			\path (\i,\j,\k) node[nodo,fill=\filling,opacity=\opnode] (nodo_\i_\j_\k) {};
		}
		\begin{scope}[wrap/.style={densely dotted,overlay}]
			\foreach \i in {0,1,2}
			\foreach \j in {0,1,2,3}
			{
				\pgfmathtruncatemacro\next{\i+1}
				\draw (nodo_\i_\j_\k) -- (nodo_\next_\j_\k);
			}
			\foreach \j in {0,1,2,3}
				\draw[wrap] (nodo_3_\j_\k) to[out=-10,in=190] (nodo_0_\j_\k);
			\foreach \i in {0,1,2,3}
			\foreach \j in {0,1,2}
			{
				\pgfmathtruncatemacro\next{\j+1}
				\draw (nodo_\i_\j_\k) -- (nodo_\i_\next_\k);
			}
			\foreach \i in {0,1,2,3}
				\draw[wrap] (nodo_\i_3_\k) to[out=80,in=-80] (nodo_\i_0_\k);
		\end{scope}
	\end{scope}
	}
	\begin{scope}[very thick,overlay,->]
	\draw[blue!100!green] (nodo_0_0_0) -- (nodo_0_0_1);
	\draw[blue!86!green] (nodo_0_0_1) -- (nodo_0_0_2);
	\draw[blue!71!green] (nodo_0_0_2) -- (nodo_0_0_3);
	\draw[blue!57!green] (nodo_0_0_3) to[out=-10,in=170] (nodo_0_2_0);
	\draw[blue!43!green] (nodo_0_2_0) -- (nodo_0_2_1);
	\draw[blue!29!green] (nodo_0_2_1) -- (nodo_0_2_2);
	\draw[blue!14!green] (nodo_0_2_2) -- (nodo_0_2_3);
	\draw[blue!00!green] (nodo_0_2_3) to[out=-10,in=170] (nodo_0_0_0);
	\end{scope}
	\draw (-1,3.5);
	\draw (7,-2.3);
	\end{tikzpicture}
	\end{center}
	\caption{The cycle $\langle \mathbf e_e\rangle$ joining by the disjoint copies of the projection.}
	\label{fig:joiningcycle}
\end{figure}}

\newcommand{\JOININGCYCLETRIDIMENSIONALSEPARATED}{
\begin{figure}[h]
	\begin{center}
	\tikzsetnextfilename{joiningcycletridimensionalseparated}
	\begin{tikzpicture}[scale=0.4,
		nodo/.style={draw,circle},
		x={(1,0)},
		y={(0,1)},
		z={(4.5,-1)},
		]
	\foreach \k in {0,1,2,3}
	{
	\foreach \i in {0,1,2,3}
	\foreach \j in {0,1,2,3}
	{
		\pgfmathtruncatemacro\cond{(\i==0)&&(\j==0 || \j==2)}
		\ifthenelse{\cond=1}{
			\def\filling{blue!40}
		}{
			\def\filling{white}
		}
		\path (\i,\j,\k) node[nodo,fill=\filling] (nodo_\i_\j_\k) {};
	}
	\begin{scope}[wrap/.style={densely dotted,overlay}]
		\foreach \i in {0,1,2}
		\foreach \j in {0,1,2,3}
		{
			\pgfmathtruncatemacro\next{\i+1}
			\draw (nodo_\i_\j_\k) -- (nodo_\next_\j_\k);
		}
		\foreach \j in {0,1,2,3}
			\draw[wrap] (nodo_3_\j_\k) to[out=-10,in=190] (nodo_0_\j_\k);
		\foreach \i in {0,1,2,3}
		\foreach \j in {0,1,2}
		{
			\pgfmathtruncatemacro\next{\j+1}
			\draw (nodo_\i_\j_\k) -- (nodo_\i_\next_\k);
		}
		\foreach \i in {0,1,2,3}
			\draw[wrap] (nodo_\i_3_\k) to[out=80,in=-80] (nodo_\i_0_\k);
	\end{scope}
	}
	\begin{scope}[very thick,overlay,->]
	\draw[blue!100!green] (nodo_0_0_0) -- (nodo_0_0_1);
	\draw[blue!86!green] (nodo_0_0_1) -- (nodo_0_0_2);
	\draw[blue!71!green] (nodo_0_0_2) -- (nodo_0_0_3);
	\draw[blue!57!green] (nodo_0_0_3) to[out=-10,in=170] (nodo_0_2_0);
	\draw[blue!43!green] (nodo_0_2_0) -- (nodo_0_2_1);
	\draw[blue!29!green] (nodo_0_2_1) -- (nodo_0_2_2);
	\draw[blue!14!green] (nodo_0_2_2) -- (nodo_0_2_3);
	\draw[blue!00!green] (nodo_0_2_3) to[out=-10,in=170] (nodo_0_0_0);
	\end{scope}
	\draw (-1.5,4);
	\draw (17.2,-3.7);
	\end{tikzpicture}
	\end{center}
	\caption{$\langle \mathbf e_n\rangle$ going by the disjoint copies of $\mathcal G(B)$.}
	\label{fig:joiningcycle}
\end{figure}}

\newcommand{\CUBICGRAPHS}{%
\begin{figure*}
	\begin{center}
	\tikzsetnextfilename{CUBICGRAPH-PC}
	\tikzset{external/export next=true}
	\begin{tikzpicture}[scale=0.8,
		nodo/.style={draw,circle,fill=white,inner sep=3pt},
		x={(0.7cm,0.2cm)},
		y={(0.0cm,0.8cm)},
		z={(-0.3cm,0.4cm)},
		]
	\foreach \j in {3,...,0}
	\foreach \i in {0,...,3}
	\foreach \k in {0,...,3}
	{
		\pgfmathtruncatemacro{\cond}{(\i==0||\i==3)+(\j==0||\j==3)+(\k==0||\k==3)>=2}
		\ifthenelse{\cond=1}{
			\def\filling{blue!40}
		}{
			\def\filling{white}
		}
		\pgfmathtruncatemacro{\cond}{((\k==3)&&(\j==0)) || (\i==0&&\j==3&&\k==0) || (\i==3&&\j==0&&\k==0)}
		\ifthenelse{\cond=1}{
			\def\line{very thick}
		}{
			\def\line{thin}
		}
		\path (\i,\k,\j) node[nodo,fill=\filling,\line] (nodo_\i_\j_\k) {};
	}
	\begin{scope}[thick,<-,overlay]
		\draw (nodo_0_0_0) to[out=180+30,in=-10] (nodo_3_0_0);
		\draw (nodo_0_0_0) to[out=-60,in=90+60] (nodo_0_3_0);
		\foreach \i in {0,...,3}
		{
			\draw (nodo_\i_0_0) to[out=-90+20,in=90-20] (nodo_\i_0_3);
		}
	\end{scope}
	\draw
		(current bounding box.south west) +(-1em,-1em)
		(current bounding box.north east) +(1em,0em)
		;
	\end{tikzpicture}
	\hspace{-1em}
	\tikzsetnextfilename{CUBICGRAPH-FCC}
	\tikzset{external/export next=true}
	\begin{tikzpicture}[scale=0.8,
		nodo/.style={draw,circle,fill=white,inner sep=3pt},
		x={(0.7cm,0.2cm)},
		y={(0.0cm,0.8cm)},
		z={(-0.3cm,0.4cm)},
		]
	\foreach \j in {3,...,0}
	\foreach \i in {0,...,7}
	\foreach \k in {0,...,3}
	{
		\pgfmathtruncatemacro{\cond}{(\i==0||\i==7)+(\j==0||\j==3)+(\k==0||\k==3)>=2}
		\ifthenelse{\cond=1}{
			\def\filling{blue!40}
		}{
			\def\filling{white}
		}
		\pgfmathtruncatemacro{\cond}{((\k==3)&&(\j==0)) || (mod(\i,4)==0&&\j==3&&\k==0) || (\i==7&&\j==0&&\k==0)}
		\ifthenelse{\cond=1}{
			\def\line{very thick}
		}{
			\def\line{thin}
		}
		\path (\i,\k,\j) node[nodo,fill=\filling,\line] (nodo_\i_\j_\k) {};
	}
	\begin{scope}[thick,<-,overlay]
		\draw (nodo_0_0_0) to[out=180+30,in=-10] (nodo_7_0_0);
		\draw (nodo_0_0_0) to[out=-60,in=90+60] (nodo_4_3_0);
		\draw (nodo_4_0_0) to[out=-60,in=90+60] (nodo_0_3_0);
		\foreach \i in {0,...,7}
		{
			\pgfmathtruncatemacro\next{mod(\i+4,8)}
			\draw (nodo_\i_0_0) to[out=-90,in=90] (nodo_\next_0_3);
		}
	\end{scope}
	\draw
		(current bounding box.south west) +(-1em,-1em)
		(current bounding box.north east) +(1em,0em)
		;
	\end{tikzpicture}
	\hspace{-1.5em}
	\tikzsetnextfilename{CUBICGRAPH-BCC}
	\tikzset{external/export next=true}
	\begin{tikzpicture}[scale=0.8,
		nodo/.style={draw,circle,fill=white,inner sep=3pt},
		x={(0.7cm,0.2cm)},
		y={(0.0cm,0.8cm)},
		z={(-0.3cm,0.4cm)},
		]
	\foreach \j in {7,...,0}
	\foreach \i in {0,...,7}
	\foreach \k in {0,...,3}
	{
		\pgfmathtruncatemacro{\cond}{(\i==0||\i==7)+(\j==0||\j==7)+(\k==0||\k==3)>=2}
		\ifthenelse{\cond=1}{
			\def\filling{blue!40}
		}{
			\def\filling{white}
		}
		\pgfmathtruncatemacro{\cond}{((\k==3)&&(\j==4)) || (\i==0&&\j==7&&\k==0) || (\i==7&&\j==0&&\k==0)}
		\ifthenelse{\cond=1}{
			\def\line{very thick}
		}{
			\def\line{thin}
		}
		\path (\i,\k,\j) node[nodo,fill=\filling,\line] (nodo_\i_\j_\k) {};
	}
	\begin{scope}[thick,<-,overlay]
		\draw (nodo_0_0_0) to[out=180+30,in=-10] (nodo_7_0_0);
		\draw (nodo_0_0_0) to[out=-60,in=90+60] (nodo_0_7_0);
		\foreach \i in {0,...,7}
		{
			\pgfmathtruncatemacro\next{mod(\i+4,8)}
			\draw (nodo_\i_0_0) to[out=-90,in=90] (nodo_\next_4_3);
		}
	\end{scope}
	\draw
		(current bounding box.south west) +(-1em,-1em)
		(current bounding box.north east) +(1em,-1em)
		;
	\end{tikzpicture}
	\end{center}
	\caption{The three cubic crystal graphs: $PC$, $FCC$ and $BCC$.}
	\label{fig:CUBICgraphs}
\end{figure*}}

\newcommand{\arbolsimetricos}{\begin{figure*}\begin{center}
	\tikzsetnextfilename{arbolsimetricos}
	\begin{tikzpicture}
		[	level 1/.style={level distance=10mm},
			level 2/.style={level distance=13mm},
			level 3/.style={level distance=15mm},
			level 4/.style={level distance=17mm},
			level 5/.style={level distance=20mm},
			level 6/.style={level distance=15mm},
			font=\tiny,
		]
	\node {$\begin{pmatrix}1\end{pmatrix}$}
		[
			sibling distance=80mm,
			edge from parent fork down,
		]
		child {
			node {$\begin{pmatrix}1&0\\0&1\end{pmatrix}$}
			[sibling distance=30mm]
			child {
				node[align=center] {$\begin{pmatrix}1&0&0\\0&1&0\\0&0&1\end{pmatrix}$\\$PC$}
				child[missing]
				child {
					node[align=center] {$\begin{pmatrix}1&0&0&0\\0&1&0&0\\0&0&1&0\\0&0&0&1\end{pmatrix}$\\$4D$-$PC$}
					child[missing]
					child {
						node[align=center] {$\begin{pmatrix}1&0&0&0&0\\0&1&0&0&0\\0&0&1&0&0\\0&0&0&1&0\\0&0&0&0&1\end{pmatrix}$\\$5D$-$PC$}
						child[missing]
						child {
							node[align=center] {$\begin{pmatrix}1&0&0&0&0&0\\0&1&0&0&0&0\\0&0&1&0&0&0\\0&0&0&1&0&0\\0&0&0&0&1&0\\0&0&0&0&0&1\end{pmatrix}$\\$6D$-$PC$}
						}
						child {
							node[align=center] {$\begin{pmatrix}2&0&0&0&0&1\\0&2&0&0&0&1\\0&0&2&0&0&1\\0&0&0&2&0&1\\0&0&0&0&2&1\\0&0&0&0&0&1\end{pmatrix}$\\$6D$-$BCC$}
						}
					}
					child {
						node[align=center] {$\begin{pmatrix}2&0&0&0&1\\0&2&0&0&1\\0&0&2&0&1\\0&0&0&2&1\\0&0&0&0&1\end{pmatrix}$\\$5D$-$BCC$}
					}
				}
				child {
					node[align=center] {$\begin{pmatrix}2&0&0&1\\0&2&0&1\\0&0&2&1\\0&0&0&1\end{pmatrix}$\\$4D$-$BCC$}
				}
			}
			child {
				node[align=center] {$\begin{pmatrix}2&0&1\\0&2&1\\0&0&1\end{pmatrix}$\\BCC}
			}
		}
		child {
			node {$\begin{pmatrix}2&1\\0&1\end{pmatrix}$}
			child {
				node[align=center] {$\begin{pmatrix}2&1&1\\0&1&0\\0&0&1\end{pmatrix}$\\FCC}
				[sibling distance=30mm]
				child {
					node[align=center] {$\begin{pmatrix}2&1&1&1\\0&1&0&0\\0&0&1&0\\0&0&0&1\end{pmatrix}$\\$4D$-$FCC$}
					child[missing]
					child[missing]
					child {
						node[align=center] {$\begin{pmatrix}2&1&1&1&1\\0&1&0&0&0\\0&0&1&0&0\\0&0&0&1&0\\0&0&0&0&1\end{pmatrix}$\\$5D$-$FCC$}
						child {
							node[align=center] {$\begin{pmatrix}2&1&1&1&1&1\\0&1&0&0&0&0\\0&0&1&0&0&0\\0&0&0&1&0&0\\0&0&0&0&1&0\\0&0&0&0&0&1\end{pmatrix}$\\$6D$-$FCC$}
						}
						child {
							node (La6) {$\begin{pmatrix}4&2&2&2&2&1\\0&2&0&0&0&1\\0&0&2&0&0&1\\0&0&0&2&0&1\\0&0&0&0&2&1\\0&0&0&0&0&1\end{pmatrix}$}
						}
						child {
							node (Lb6) {$\begin{pmatrix}4&2&2&2&2&3\\0&2&0&0&0&1\\0&0&2&0&0&1\\0&0&0&2&0&1\\0&0&0&0&2&1\\0&0&0&0&0&1\end{pmatrix}$}
						}
					}
				}
				child {
					node[align=center] (La) {$\begin{pmatrix}4&2&2&1\\0&2&0&1\\0&0&2&1\\0&0&0&1\end{pmatrix}$\\$Lip$}
				}
				child {
					node (Lb) {$\begin{pmatrix}4&2&2&3\\0&2&0&1\\0&0&2&1\\0&0&0&1\end{pmatrix}$}
				}
			}
		}
		;
		\path (La) -- node[midway] {$\simeq$} (Lb);
		\path (La6) -- node[midway] {$\simeq$} (Lb6);
\end{tikzpicture}\caption{Tree showing lifts and projections of cubic crystal graphs up to dimension 6.}
	\label{fig:arbol}
\end{center}\end{figure*}}

\newcommand{\BGQPARTITIONS}{%
\begin{figure}
	\begin{center}
	\tikzsetnextfilename{BGQpartition4}
	\begin{tikzpicture}
	\expandafter\def\csname mode0\endcsname{0}
	\expandafter\def\csname mode1\endcsname{0}
	\expandafter\def\csname mode2\endcsname{0}
	\expandafter\def\csname mode3\endcsname{0}
	\foreach \i in {0,1,2,3}
	{
		\draw (2*\i,0) rectangle ++(1,1);
		\ifthenelse{\csname mode\i\endcsname=0}{
			\draw (2*\i+.3,0) -- ++(0,1);
			\draw (2*\i+.7,0) -- ++(0,1);
		}
		{
			\draw (2*\i+.3,0) to[overlay,out=90,in=90,looseness=2.0] ++(.4,0);
			\draw (2*\i+.3,1) to[overlay,out=-90,in=-90,looseness=2.0] ++(.4,0);
		}
	}
	\foreach \i in {0,1,2}
		\draw (2*\i+.7,1) -- ++(0,.2) -| (2*\i+2.3,1);
	\draw (6.7,1) -- ++(0,.4) -| (0.3,1);
	\end{tikzpicture}\\
	\vspace{1ex}
	\tikzsetnextfilename{BGQpartition31}
	\begin{tikzpicture}
	\expandafter\def\csname mode0\endcsname{0}
	\expandafter\def\csname mode1\endcsname{0}
	\expandafter\def\csname mode2\endcsname{0}
	\expandafter\def\csname mode3\endcsname{1}
	\foreach \i in {0,1,2,3}
	{
		\draw (2*\i,0) rectangle ++(1,1);
		\ifthenelse{\csname mode\i\endcsname=0}{
			\draw (2*\i+.3,0) -- ++(0,1);
			\draw (2*\i+.7,0) -- ++(0,1);
		}
		{
			\draw (2*\i+.3,0) to[overlay,out=90,in=90,looseness=2.0] ++(.4,0);
			\draw (2*\i+.3,1) to[overlay,out=-90,in=-90,looseness=2.0] ++(.4,0);
		}
	}
	\foreach \i in {0,1,2}
		\draw (2*\i+.7,1) -- ++(0,.2) -| (2*\i+2.3,1);
	\draw (6.7,1) -- ++(0,.4) -| (0.3,1);
	\end{tikzpicture}\\
	\vspace{1ex}
	\tikzsetnextfilename{BGQpartition211}
	\begin{tikzpicture}
	\expandafter\def\csname mode0\endcsname{0}
	\expandafter\def\csname mode1\endcsname{0}
	\expandafter\def\csname mode2\endcsname{1}
	\expandafter\def\csname mode3\endcsname{1}
	\foreach \i in {0,1,2,3}
	{
		\draw (2*\i,0) rectangle ++(1,1);
		\ifthenelse{\csname mode\i\endcsname=0}{
			\draw (2*\i+.3,0) -- ++(0,1);
			\draw (2*\i+.7,0) -- ++(0,1);
		}
		{
			\draw (2*\i+.3,0) to[overlay,out=90,in=90,looseness=2.0] ++(.4,0);
			\draw (2*\i+.3,1) to[overlay,out=-90,in=-90,looseness=2.0] ++(.4,0);
		}
	}
	\foreach \i in {0,1,2}
		\draw (2*\i+.7,1) -- ++(0,.2) -| (2*\i+2.3,1);
	\draw (6.7,1) -- ++(0,.4) -| (0.3,1);
	\end{tikzpicture}\\
	\vspace{1ex}
	\tikzsetnextfilename{BGQpartition1111}
	\begin{tikzpicture}
	\expandafter\def\csname mode0\endcsname{1}
	\expandafter\def\csname mode1\endcsname{1}
	\expandafter\def\csname mode2\endcsname{1}
	\expandafter\def\csname mode3\endcsname{1}
	\foreach \i in {0,1,2,3}
	{
		\draw (2*\i,0) rectangle ++(1,1);
		\ifthenelse{\csname mode\i\endcsname=0}{
			\draw (2*\i+.3,0) -- ++(0,1);
			\draw (2*\i+.7,0) -- ++(0,1);
		}
		{
			\draw (2*\i+.3,0) to[overlay,out=90,in=90,looseness=2.0] ++(.4,0);
			\draw (2*\i+.3,1) to[overlay,out=-90,in=-90,looseness=2.0] ++(.4,0);
		}
	}
	\foreach \i in {0,1,2}
		\draw (2*\i+.7,1) -- ++(0,.2) -| (2*\i+2.3,1);
	\draw (6.7,1) -- ++(0,.4) -| (0.3,1);
	\end{tikzpicture}
	\end{center}
	\caption{The 4 partitions available in the BG/Q}
	\label{fig:BGQpartitions}
\end{figure}}

\newcommand{\RTTMIDPLANES}{
\begin{figure}
	\begin{center}
	\tikzsetnextfilename{RTTmidplanes}
	\tikzset{external/export next=true}
	\begin{tikzpicture}[
		scale=0.7,
		estilo 1/.style={blue},
		estilo 2/.style={black},
		]
	\foreach \j in {0,1,2}
	{
		\foreach \i in {0,1,2,3}
		{
			\begin{scope}[estilo 1]
			\pgfmathtruncatemacro\cond{(\j==1)*(\i>=1)*(\i<=2)}
			\ifthenelse{\cond=1}{
				\xdef\filling{red!10}
				\xdef\line{thick}
			}{
				\xdef\filling{none}
				\xdef\line{thin}
			}
			\draw[fill=\filling] (3*\i,2*\j) coordinate (lc1) rectangle ++(1,1);
				\draw[\line] (lc1) ++(.3,0) coordinate (lc1_\i_\j_m0) -- ++(0,1) coordinate (lc1_\i_\j_o0);
				\draw[\line] (lc1) ++(.7,0) coordinate (lc1_\i_\j_m1) -- ++(0,1) coordinate (lc1_\i_\j_o1);
			\end{scope}
			\begin{scope}[estilo 2]
			\draw[fill=\filling] (1.1+3*\i,2*\j) coordinate (lc2) rectangle ++(1,1);
				\draw[] (lc2) ++(.3,0) coordinate (lc2_\i_\j_m0) -- ++(0,1) coordinate (lc2_\i_\j_o0);
				\draw[] (lc2) ++(.7,0) coordinate (lc2_\i_\j_m1) -- ++(0,1) coordinate (lc2_\i_\j_o1);
				\draw[\line] (lc2) ++(.3,0) -- ++(-.2,1) coordinate (lc2_\i_\j_t0);
				\draw[\line] (lc2) ++(.7,0) -- ++(.2,1) coordinate (lc2_\i_\j_t1);
			\end{scope}
		}
		\ifthenelse{\j=1}{
			\xdef\line{thick}
		}{
			\xdef\line{thin}
		}
		\foreach \i in {0,1,2}
		{
			\pgfmathtruncatemacro\next{\i+1}
			\draw[estilo 1,\line] (lc1_\i_\j_o1) -- ++(0,.1) -| (lc1_\next_\j_o0);
		}
		\draw[estilo 1,\line] (lc1_3_\j_o1) -- ++(0,.2) -| (lc1_0_\j_o0);
	}
	\foreach \i in {0,1,2}
	{
		\pgfmathtruncatemacro\next{\i+1}
		\draw[estilo 2,thick] (lc2_\i_1_t1) -- ++(0,.3) -| (lc2_\next_1_t0);
	}
	\draw[estilo 2,thick] (lc2_0_1_t0) -- ++(0,.5) -| (lc2_3_1_t1);
	\foreach \i in {0,3}
	{
		\draw[estilo 2,thick] (lc2_\i_1_t0) to[out=-90,in=-90] (lc2_\i_1_t1);
		\foreach \j in {0,1}
		{
			\pgfmathtruncatemacro\next{\j+1}
			\draw[estilo 2] (lc2_\i_\j_o0) |- ++(.8,.7) -- ++(0,1.7) -| (lc2_\i_\next_o1);
		}
	}
	\draw[estilo 1,thick] (lc1_0_1_o1) to[out=-90,in=-90] (lc1_0_1_o0);
	\draw[estilo 1,thick] (lc1_3_1_o1) to[out=-90,in=-90] (lc1_3_1_o0);
	\end{tikzpicture}
	\end{center}
	\caption{Building a RTT of two midplanes}
	\label{fig:RTT2midplanes}
\end{figure}}

\newcommand{\RTTHMIDPLANES}{
\begin{figure}
	\begin{center}
	\tikzsetnextfilename{RTT8midplanes}
	\tikzset{external/export next=true}
	\begin{tikzpicture}[
		scale=0.7,
		estilo 1/.style={blue},
		estilo 2/.style={black},
		]
	\foreach \j in {0,1,2,3}
	{
		\foreach \i in {0,1,2,3}
		{
			\begin{scope}[estilo 1]
			\pgfmathtruncatemacro\cond{(\j>=1)&&(\j<=2)}
			\ifthenelse{\cond=1}{
				\xdef\filling{red!10}
				\xdef\line{thick}
				\ifthenelse{\j=1}{
					\xdef\linem{thick}
					\xdef\linet{thin}
				}{
					\xdef\linem{thin}
					\xdef\linet{thick}
				}
			}{
				\xdef\filling{none}
				\xdef\line{thin}
				\xdef\linem{thin}
				\xdef\linet{thin}
			}
			\draw[fill=\filling] (3*\i,2*\j) coordinate (lc1) rectangle ++(1,1);
				\draw[\line] (lc1) ++(.3,0) coordinate (lc1_\i_\j_m0) -- ++(0,1) coordinate (lc1_\i_\j_o0);
				\draw[\line] (lc1) ++(.7,0) coordinate (lc1_\i_\j_m1) -- ++(0,1) coordinate (lc1_\i_\j_o1);
			\end{scope}
			\begin{scope}[estilo 2]
			\draw[fill=\filling] (1.1+3*\i,2*\j) coordinate (lc2) rectangle ++(1,1);
				\draw[\linem] (lc2) ++(.3,0) coordinate (lc2_\i_\j_m0) -- ++(0,1) coordinate (lc2_\i_\j_o0);
				\draw[\line] (lc2) ++(.7,0) coordinate (lc2_\i_\j_m1) -- ++(0,1) coordinate (lc2_\i_\j_o1);
				\draw[\linet] (lc2) ++(.3,0) -- ++(-.2,1) coordinate (lc2_\i_\j_t0);
				\draw[] (lc2) ++(.7,0) -- ++(.2,1) coordinate (lc2_\i_\j_t1);
			\end{scope}
		}
		\pgfmathtruncatemacro\cond{(\j>=1)&&(\j<=2)}
		\ifthenelse{\cond=1}{
			\xdef\line{thick}
		}{
			\xdef\line{thin}
		}
		\foreach \i in {0,1,2}
		{
			\pgfmathtruncatemacro\next{\i+1}
			\draw[estilo 1,\line] (lc1_\i_\j_o1) -- ++(0,.1) -| (lc1_\next_\j_o0);
		}
		\draw[estilo 1,\line] (lc1_3_\j_o1) -- ++(0,.2) -| (lc1_0_\j_o0);
	}
	\foreach \i in {0,1,2,3}
	{
		\draw[estilo 2,thick] (lc2_\i_2_t1) to[out=-90,in=-90] (lc2_\i_2_o0);
		\draw[estilo 2,thick] (lc2_\i_3_o1) to[out=-90,in=-90,looseness=2.0] (lc2_\i_3_t0);
		\draw[estilo 2,thick] (lc2_\i_3_o0) to[out=-90,in=-90,looseness=3.0] (lc2_\i_3_t1);
		\draw[estilo 2,thick] (lc2_\i_0_o1) to[out=-90,in=-90] (lc2_\i_0_o0);
		\draw[estilo 2,thick] (lc2_\i_3_o0) |- ++(1,1.2) -- ++(0,-6.7) -| (lc2_\i_0_o1);
	}
	\foreach \j in {2,3}
	{
		\foreach \i in {0,1,2}
		{
			\pgfmathtruncatemacro\next{\i+1}
			\draw[estilo 2,thick] (lc2_\i_\j_t1) -- ++(0,.3) -| (lc2_\next_\j_t0);
		}
		\draw[estilo 2,thick] (lc2_0_\j_t0) -- ++(0,.5) -| (lc2_3_\j_t1);
	}
	\foreach \i in {0,1,2,3}
	{
		\foreach \j in {0,1,2}
		{
			\pgfmathtruncatemacro\next{\j+1}
			\draw[estilo 2,thick] (lc2_\i_\j_o0) |- ++(.8,.7) -- ++(0,1.7) -| (lc2_\i_\next_o1);
		}
	}
	\end{tikzpicture}
	\end{center}
	\caption{Building a RTT of two midplanes}
	\label{fig:RTT8midplanes}
\end{figure}}

\newcommand{\CRAYPCPBCC}{
\begin{figure}
	\begin{center}
	\tikzsetnextfilename{crayPCpBCC}
	\tikzset{external/export next=true}
	\begin{tikzpicture}[
		rack/.style={draw,regular polygon,regular polygon sides=4},
		nodo/.style={draw,circle,minimum size=1pt,inner sep=0pt},
		]
	\begin{scope}[scale=0.7]
		\foreach \i in {0,...,7}
		\foreach \j in {0,...,3}
		{
			\path (\i,\j) node[rack] (rack_\i_\j) {};
		}
		\foreach \i in {0,...,7}
		\foreach \j in {0,...,7}
		{
			\draw (\i*1.6pt-5.7pt,\j*1.6pt-5.8pt) circle (.05pt);
			\draw (\i*1.6pt-5.7pt,3cm+\j*1.6pt-5.8pt) circle (.05pt);
		}
		\begin{scope}[overlay,->]
			\foreach \i in {0,...,6}
			\foreach \j in {0,...,3}
			{
				\pgfmathtruncatemacro\next{\i+1}
				\draw[blue] (rack_\i_\j) -- (rack_\next_\j);
			}
			\foreach \j in {0,...,3}
				\draw[blue] (rack_7_\j) to[out=10,in=180-10] (rack_0_\j);
			\foreach \i in {0,...,7}
			\foreach \j in {0,...,2}
			{
				\pgfmathtruncatemacro\next{\j+1}
				\draw[red] (rack_\i_\j) -- (rack_\i_\next);
			}
			\foreach \i in {0,...,7}
			{
				\draw[red,decoration={markings,mark=at position 0.95 with {\arrow[thick,scale=0.5]{|}}},postaction=decorate] (rack_\i_3) to[out=90+10,in=-90-10,looseness=1.5] (rack_\i_0);
			}
		\end{scope}
		\draw
			(current bounding box.south west) +(-1.5em,-1.2em)
			(current bounding box.north east) +(1em,1em)
			;
	\end{scope}
	\begin{scope}[xshift=6.0cm,scale=0.2,
		]
		\begin{scope}
			\foreach \i in {0,...,7}
			\foreach \j in {0,...,7}
			{
				\pgfmathtruncatemacro\cond{(mod(\i,4)==0)&&(mod(\j,4)==0)}
				\ifthenelse{\cond=1}{
					\def\line{ultra thick}
				}{
					\def\line{thin}
				}
				\path (\i,\j) node[nodo,\line] (nodo_0_\i_\j) {};
			}
		\end{scope}
		\node at (3.5,9.5) {$\vdots$};
		\begin{scope}[yshift=11cm]
			\foreach \i in {0,...,7}
			\foreach \j in {0,...,7}
			{
				\pgfmathtruncatemacro\cond{(mod(\i,4)==0)&&(mod(\j,4)==0)}
				\ifthenelse{\cond=1}{
					\def\line{ultra thick}
				}{
					\def\line{thin}
				}
				\path (\i,\j) node[nodo,\line] (nodo_3_\i_\j) {};
			}
		\end{scope}
		\begin{scope}[overlay,thick,red,->]
			\draw (nodo_3_4_4) to[out=90,in=-90] (nodo_0_0_0);
			\draw (nodo_3_0_0) to[out=90,in=-90] (nodo_0_4_4);
			\draw[dashed,thin] (nodo_3_4_0) to[out=90,in=-90] (nodo_0_0_4);
			\draw[dashed,thin] (nodo_3_0_4) to[out=90,in=-90] (nodo_0_4_0);
		\end{scope}
	\end{scope}
	\end{tikzpicture}
	\end{center}
	\caption{Cray-like physical layout of $PC(8)\boxplus BCC(4)$.}
	\label{fig:crayPCpBCC}
\end{figure}}

\newcommand{\ROOMPARTITIONS}{%
\begin{figure*}
	\begin{center}
	\tikzsetnextfilename{roompartitions}
	\tikzset{external/export next=true}
	\begin{tikzpicture}[
		label/.style={font=\scriptsize,anchor=base},
		partition -1/.style={fill=green!20},
		partition 0/.style={fill=red},
		partition 1/.style={fill=yellow},
		partition 2/.style={fill=blue},
		partition 3/.style={fill=purple},
		partition 4/.style={fill=gray},
		partition 5/.style={fill=blue!40},
		partition 6/.style={fill=red!40},
		partition 7/.style={fill=blue!30},
		partition 8/.style={fill=green!40!blue},
		partition 9/.style={fill=gray},
		mode label/.style={font=\tiny},
		mode -1/.style={fill=white},
		mode 0/.style={fill=blue},
		mode 1/.style={fill=green},
		mode 2/.style={fill=red},
		mode 3/.style={fill=blue!50!green},
		mode 4/.style={fill=blue!50!red},
		mode 5/.style={fill=green!50!red},
		mode 6/.style={fill=blue!50!black},
		mode 7/.style={fill=green!50!black},
		mode 8/.style={fill=red!50!black},
		mode 9/.style={fill=black!80},
		modeline/.style={thick},
		modeline -1/.style={modeline,draw=white},
		modeline 0/.style={modeline,draw=blue},
		modeline 1/.style={modeline,draw=green},
		modeline 2/.style={modeline,draw=red},
		modeline 3/.style={modeline,draw=blue!50!green},
		modeline 4/.style={modeline,draw=blue!50!red},
		modeline 5/.style={modeline,draw=green!50!red},
		modeline 6/.style={modeline,draw=blue!50!black},
		modeline 7/.style={modeline,draw=green!50!black},
		modeline 8/.style={modeline,draw=red!50!black},
		modeline 9/.style={modeline,draw=black!80},
		]
	\pgfmathsetlengthmacro\rackheight{1cm}
	\pgfmathsetlengthmacro\midplaneheight{.5*\rackheight}
	\pgfmathsetlengthmacro\rackwidth{.5cm}
	\pgfmathsetlengthmacro\racksep{.1cm}
	\pgfmathsetlengthmacro\rowsep{.4cm}
	\pgfmathsetlengthmacro\rackint{\rackheight+\racksep}
	\pgfmathsetlengthmacro\rowint{\rackwidth+\rowsep}
	\pgfmathsetlengthmacro\lognodeheight{.3cm}
	\pgfmathsetlengthmacro\modewidth{\rackwidth/4}
	\foreach \logx in {0,1,2,3}
	\foreach \logy in {0,1,2,3}
	\foreach \logz in {0,1,2,3}
	\foreach \logw in {0,1,2}
	{
		\expandafter\xdef\csname partition_\logx_\logy_\logz_\logw\endcsname{-1}
		\foreach \d in {x,y,z,w}
		{
			\expandafter\xdef\csname mode_\logx_\logy_\logz_\logw_\d\endcsname{-1}
		}
	}
	\expandafter\xdef\csname partition_0_1_0_0\endcsname{1}
	\expandafter\xdef\csname mode_0_1_0_0_x\endcsname{1}
	\foreach \row in {0,...,11}
	{
		\pgfmathtruncatemacro\phyx{mod(\row,4)}
		\pgfmathtruncatemacro\phyw{(\row-\phyx)/4}
		\pgfmathtruncatemacro\logx{mod(\phyx*(\phyx+1)/2,4)}
		\pgfmathtruncatemacro\logw{\phyw}
		\foreach \rack in {0,...,7}
		{
			\pgfmathtruncatemacro\phyry{mod(\rack,2)}
			\pgfmathtruncatemacro\phyz{(\rack-\phyry)/2}
			\pgfmathtruncatemacro\logz{mod(\phyz*(\phyz+1)/2,4)}
			\path (\row*\rowint,\rack*\rackint) coordinate (rack_\row_\rack);
			\pgfmathtruncatemacro\phyy{2*\phyry}
			\pgfmathtruncatemacro\logy{mod(\phyy*(\phyy+1)/2,4)}
			\edef\partition{\csname partition_\logx_\logy_\logz_\logw\endcsname}
			\path (rack_\row_\rack)
				coordinate (phy_\phyx_\phyy_\phyz_\phyw)
				coordinate (log_\logx_\logy_\logz_\logw)
				;
			\path[draw,partition \partition]
				(rack_\row_\rack) rectangle +(\rackwidth,\midplaneheight);
			\path (rack_\row_\rack) +(.5*\rackwidth,\lognodeheight) node[label] {\logx \logy \logz \logw};
			\foreach \i/\d in {0/x,1/y,2/z,3/w}
			{
				\edef\mode{\csname mode_\logx_\logy_\logz_\logw_\d\endcsname}
				\draw[mode \mode] (rack_\row_\rack) ++(\i*\modewidth+.5pt,1pt) rectangle +(\modewidth-.5pt,\modewidth-.5pt) +(.5*\modewidth,.5*\modewidth)
					;
			}
			\pgfmathtruncatemacro\phyy{2*\phyry+1}
			\pgfmathtruncatemacro\logy{mod(\phyy*(\phyy+1)/2,4)}
			\edef\partition{\csname partition_\logx_\logy_\logz_\logw\endcsname}
			\draw (rack_\row_\rack) ++(0,\midplaneheight)
				coordinate (rackH_\row_\rack)
				coordinate (phy_\phyx_\phyy_\phyz_\phyw)
				coordinate (log_\logx_\logy_\logz_\logw)
				-- +(\rackwidth,0)
				;
			\path[draw,partition \partition]
				(rackH_\row_\rack) rectangle +(\rackwidth,\rackheight-\midplaneheight);
			\path (rackH_\row_\rack) +(.5*\rackwidth,\lognodeheight) node[label] {\logx \logy \logz \logw};
			\foreach \i/\d in {0/x,1/y,2/z,3/w}
			{
				\edef\mode{\csname mode_\logx_\logy_\logz_\logw_\d\endcsname}
				\draw[mode \mode] (rackH_\row_\rack) ++(\i*\modewidth+.5pt,1pt) rectangle +(\modewidth-.5pt,\modewidth-.5pt) +(.5*\modewidth,.5*\modewidth)
					;
			}
		}
	}
	\begin{scope}[shift={(11cm,9cm)},
			every node/.style={anchor=west,font=\scriptsize},
		]
		\newcounter{legendline}
		\pgfmathsetcounter{legendline}{0}
		\draw[] (0,-\thelegendline*1em) node {Partitions};
		\pgfmathaddtocounter{legendline}{1}
		\draw[partition 0] (0,-\thelegendline*1em) +(-1em,-.5em) rectangle +(0,.5em) +(0,0) node {Single midplane partitions};
		\pgfmathaddtocounter{legendline}{1}
		\draw[partition 1] (0,-\thelegendline*1em) +(-1em,-.5em) rectangle +(0,.5em) +(0,0) node {$4FCC(2)$};
		\pgfmathaddtocounter{legendline}{1}
		\draw[partition 2] (0,-\thelegendline*1em) +(-1em,-.5em) rectangle +(0,.5em) +(0,0) node {First $PC(2)\boxplus BCC(1)$};
		\pgfmathaddtocounter{legendline}{1}
		\draw[partition 3] (0,-\thelegendline*1em) +(-1em,-.5em) rectangle +(0,.5em) +(0,0) node {Second $PC(2)\boxplus BCC(1)$};
		\pgfmathaddtocounter{legendline}{1}
		\draw[partition 4] (0,-\thelegendline*1em) +(-1em,-.5em) rectangle +(0,.5em) +(0,0) node {4PC(2)};
		\pgfmathaddtocounter{legendline}{1}
		\draw[partition 5] (0,-\thelegendline*1em) +(-1em,-.5em) rectangle +(0,.5em) +(0,0) node {$4FCC(1)$};
		\pgfmathaddtocounter{legendline}{1}
		\draw[partition 6] (0,-\thelegendline*1em) +(-1em,-.5em) rectangle +(0,.5em) +(0,0) node {$4FCC(1)$};
		\pgfmathaddtocounter{legendline}{1}
		\draw[partition 7] (0,-\thelegendline*1em) +(-1em,-.5em) rectangle +(0,.5em) +(0,0) node {$4FCC(1)$};
		\pgfmathaddtocounter{legendline}{1}
		\draw[partition 8] (0,-\thelegendline*1em) +(-1em,-.5em) rectangle +(0,.5em) +(0,0) node {$4FCC(1)$};
		\pgfmathaddtocounter{legendline}{2}
		\draw[] (0,-\thelegendline*1em) node {Modes};
		\pgfmathaddtocounter{legendline}{1}
		\foreach \i in {0,...,9}
		{
			\pgfmathaddtocounter{legendline}{\i}
			\draw[mode \i] (0,-\thelegendline*1em) +(-1em,-.5em) rectangle +(0,.5em) +(0,0) node {Mode \i};
		}
	\end{scope}
	\pgfmathsetlengthmacro\lcwidth{2cm}
	\pgfmathsetlengthmacro\sep{.8cm}
	\pgfmathsetlengthmacro\portradius{1pt}
	\begin{scope}[shift={(-3cm,8cm)},
			port/.style={font=\tiny,inner sep=.2pt},
			]
		\begin{scope}[shift={(0cm,0cm)},]
			\draw (0,0) rectangle +(\lcwidth,\lcwidth);
			\coordinate (ml) at (\lcwidth/3,0);
			\coordinate (mr) at (2*\lcwidth/3,0);
			\coordinate (p1l) at (\lcwidth/3,\lcwidth);
			\coordinate (p1r) at (2*\lcwidth/3,\lcwidth);
			\draw[modeline 0] (ml) to[out=90,in=90] (mr) (p1l) to[out=-90,in=-90] (p1r);
			\draw[modeline 1] (ml) -- (p1l) (mr) -- (p1r);
			\foreach \port in {ml,mr,p1l,p1r}
				\fill (\port) circle (\portradius);
			\node[port,anchor=south] at (p1l) {$+\mathbf e_1$};
			\node[port,anchor=south] at (p1r) {$-\mathbf e_1$};
		\end{scope}
		\begin{scope}[shift={(0cm,-\lcwidth-\sep)},]
			\draw (0,0) rectangle +(\lcwidth,\lcwidth);
			\coordinate (ml) at (\lcwidth/3,0);
			\coordinate (mr) at (2*\lcwidth/3,0);
			\coordinate (p1l) at (\lcwidth/3,\lcwidth);
			\coordinate (p1r) at (2*\lcwidth/3,\lcwidth);
			\coordinate (p2l) at (\lcwidth/5,\lcwidth);
			\coordinate (p2r) at (4*\lcwidth/5,\lcwidth);
			\draw[modeline 0] (ml) to[out=90,in=90] (mr) (p1l) to[out=-90,in=-90] (p1r) (p2l) to[out=-90,in=-90] (p2r);
			\draw[modeline 1] (ml) -- (p1l) (mr) -- (p1r);
			\draw[modeline 2] (ml) -- (p2l) (mr) -- (p2r);
			\draw[modeline 4,looseness=0.5] (ml) to[bend left] (p2l) (mr) to[bend left] (p1r);
			\draw[modeline 5,looseness=2.0] (ml) to[out=90,in=90] (mr) (p1l) to[out=-90,in=-90] (p2r) (p2l) to[out=-90,in=-90] (p1r);
			\foreach \port in {ml,mr,p1l,p1r,p2l,p2r}
				\fill (\port) circle (\portradius);
			\node[port,anchor=south] at (p1l) {$+\mathbf e_2$};
			\node[port,anchor=south] at (p1r) {$-\mathbf e_2$};
			\node[port,anchor=south east] at (p2l) {$+\begin{pmatrix}1\\1\\0\\0\end{pmatrix}$};
			\node[port,anchor=south west] at (p2r) {$-\begin{pmatrix}1\\1\\0\\0\end{pmatrix}$};
		\end{scope}
		\begin{scope}[shift={(0cm,-2*\lcwidth-2*\sep)},]
			\draw (0,0) rectangle +(\lcwidth,\lcwidth);
			\coordinate (ml) at (\lcwidth/3,0);
			\coordinate (mr) at (2*\lcwidth/3,0);
			\coordinate (p1l) at (\lcwidth/3,\lcwidth);
			\coordinate (p1r) at (2*\lcwidth/3,\lcwidth);
			\coordinate (p2l) at (\lcwidth/5,\lcwidth);
			\coordinate (p2r) at (4*\lcwidth/5,\lcwidth);
			\draw[modeline 0] (ml) to[out=90,in=90] (mr) (p1l) to[out=-90,in=-90] (p1r) (p2l) to[out=-90,in=-90] (p2r);
			\draw[modeline 1] (ml) -- (p1l) (mr) -- (p1r);
			\draw[modeline 2] (ml) -- (p2l) (mr) -- (p2r);
			\draw[modeline 4,looseness=0.5] (ml) to[bend left] (p2l) (mr) to[bend left] (p1r);
			\draw[modeline 5,looseness=2.0] (ml) to[out=90,in=90] (mr) (p1l) to[out=-90,in=-90] (p2r) (p2l) to[out=-90,in=-90] (p1r);
			\foreach \port in {ml,mr,p1l,p1r,p2l,p2r}
				\fill (\port) circle (\portradius);
			\node[port,anchor=south] at (p1l) {$+\mathbf e_3$};
			\node[port,anchor=south] at (p1r) {$-\mathbf e_3$};
			\node[port,anchor=south east] at (p2l) {$+\begin{pmatrix}1\\0\\1\\0\end{pmatrix}$};
			\node[port,anchor=south west] at (p2r) {$-\begin{pmatrix}1\\0\\1\\0\end{pmatrix}$};
		\end{scope}
		\begin{scope}[shift={(0cm,-3*\lcwidth-3*\sep)},]
			\draw (0,0) rectangle +(\lcwidth,\lcwidth);
			\coordinate (ml) at (\lcwidth/3,0);
			\coordinate (mr) at (2*\lcwidth/3,0);
			\coordinate (p1l) at (\lcwidth/3,\lcwidth);
			\coordinate (p1r) at (2*\lcwidth/3,\lcwidth);
			\coordinate (p2l) at (\lcwidth/5,\lcwidth);
			\coordinate (p2r) at (4*\lcwidth/5,\lcwidth);
			\coordinate (p3l) at (0,2*\lcwidth/3);
			\coordinate (p3r) at (\lcwidth,2*\lcwidth/3);
			\draw[modeline 0] (ml) to[out=90,in=90] (mr) (p1l) to[out=-90,in=-90] (p1r) (p2l) to[out=-90,in=-90] (p2r) (p3l) to[] (p3r);
			\draw[modeline 1] (ml) -- (p1l) (mr) -- (p1r);
			\draw[modeline 2] (ml) -- (p2l) (mr) -- (p2r);
			\draw[modeline 3] (ml) -- (p3l) (mr) -- (p3r);
			\draw[modeline 4,looseness=0.5] (ml) to[bend left] (p2l) (mr) to[bend left] (p1r);
			\draw[modeline 5,looseness=2.0] (ml) to[out=90,in=90] (mr) (p1l) to[out=-90,in=-90] (p2r) (p2l) to[out=-90,in=-90] (p1r);
			\draw[modeline 6,looseness=0.5] (ml) to[bend right] (p3l) (mr) to[bend right] (p1r);
			\draw[modeline 7] (ml) to[out=90,in=90,looseness=3.0] (mr) (p1l) to[out=-90,in=180] (p3r) (p3l) to[out=0,in=-90] (p1r);
			\foreach \port in {ml,mr,p1l,p1r,p2l,p2r,p3l,p3r}
				\fill (\port) circle (\portradius);
			\node[port,anchor=south] at (p1l) {$+\mathbf e_4$};
			\node[port,anchor=south] at (p1r) {$-\mathbf e_4$};
			\node[port,anchor=south east] at (p2l) {$+\begin{pmatrix}1\\0\\0\\1\end{pmatrix}$};
			\node[port,anchor=south west] at (p2r) {$-\begin{pmatrix}1\\0\\0\\1\end{pmatrix}$};
			\node[port,anchor=east] at (p3l) {$+\begin{pmatrix}1\\1\\1\\1\end{pmatrix}$};
			\node[port,anchor=west] at (p3r) {$-\begin{pmatrix}1\\1\\1\\1\end{pmatrix}$};
		\end{scope}
	\end{scope}
	\end{tikzpicture}
	\end{center}
	\caption{Physical layout and partitioning example}
	\label{fig:roompartitions}
\end{figure*}}

\newcommand{\SIMULATIONLATENCYLARGE}{%
\begin{figure}[th]
        \begin{center}
        \pgfplotsset{minor grid style={dashed,very thin, color=blue!15}}
        \pgfplotsset{major grid style={very thin, color=black!30}}
		\tikzsetnextfilename{simulation_latency_large}
        \begin{tikzpicture}
        \begin{axis}[
                xmin=0,xmax=0.8,
                ymin=0.0,
                ymax=900,
                ymajorgrids=true,
                yminorgrids=true,
                xmajorgrids=true,
                minor y tick num=4,
				xlabel={offered load ($\text{phits}/(\text{cycle}\cdot \text{node})$)},
				ylabel={average latency (cycles)},
                legend style={at={(1.00,1.01)},anchor=south east,font=\scriptsize},legend columns=2,legend cell align=left,
                ,
        ]

\addplot+[blue,solid,mark=o,mark options=solid] coordinates{(0.02,24.695)
(0.1,29.079742)(0.2,36.4860046)(0.3,47.0441354)(0.4,62.8486236)(0.5,88.053036)(0.6,134.2232782)(0.7,289.2993822)(0.8,1171.6512962)(0.9,1698.5279634)(1.0,1878.3378614)(1.1,1945.343389)(1.2,1973.3136684)(1.3,1988.5523786)(1.4,1997.6643912)(1.5,2004.4572614)(1.6,2007.8809088)(1.7,2010.0002684)(1.8,2011.3303242)(1.9,2012.3903098)(2.0,2013.0015774)};\addlegendentry{$4D$-$FCC$(8)+uniform};
\addplot+[blue,dashed,mark=o,mark options=solid] coordinates{(0.02,26.1518962)
(0.1,31.9529887)(0.2,43.3177318)(0.3,64.9426984)(0.4,123.933218)(0.5,891.5204562)(0.6,2248.3947908)(0.7,2509.5093834)(0.8,2549.3747854)(0.9,2563.944673)(1.0,2572.341999)(1.1,2576.4720856)(1.2,2579.3965542)(1.3,2581.3880512)(1.4,2584.344393)(1.5,2584.493971)(1.6,2584.4885102)(1.7,2584.7450396)(1.8,2585.3771428)(1.9,2586.47343)(2.0,2587.4813677)};\addlegendentry{T(16,8,8,8)+uniform};
\addplot+[red,solid,mark=*,mark options=solid] coordinates{(0.02,33.2231758)
(0.1,45.4534138)(0.2,73.581581)(0.3,138.5018938)(0.4,461.1791028)(0.5,4086.2117462)(0.6,5331.5164792)(0.7,5925.193704)(0.8,6284.5860458)(0.9,6642.1250204)(1.0,6778.4393586)(2.0,7120.8433656)};\addlegendentry{$4D$-$FCC$(8)+antipodal};
\addplot+[red,dashed,mark=*,mark options=solid] coordinates{(0.02,37.956361)
(0.1,57.7823765)(0.2,160.5057232)(0.3,2306.8336672)(0.4,3822.3849934)(0.5,4547.8627566)(0.6,4780.7095742)(0.7,4874.6792542)(0.8,4912.5259886)(0.9,4900.8907506)(1.0,4985.8927416)(2.0,4990.1386606)};\addlegendentry{T(16,8,8,8)+antipodal};
\addplot+[green,solid,mark=square*,mark options=solid] coordinates{(0.02,25.165077)
(0.1,30.941191)(0.2,45.5706246)(0.3,271.7828698)(0.4,727.3442946)(0.5,1231.981991)(0.6,2134.9221936)(0.7,1742.6848854)(0.8,2065.9972234)(0.9,2205.4859136)(1.0,2256.939519)(1.1,2396.2041194)(1.2,2602.8534038)(1.3,2695.3978144)(1.4,2698.0692176)(1.5,2718.346883)(1.6,2718.5390076)(1.7,2717.8835724)(1.8,2725.2471132)(1.9,2711.187059)(2.0,2733.1961478)};\addlegendentry{$4D$-$FCC$(8)+centralsymmetric};
\addplot+[green,dashed,mark=square*,mark options=solid] coordinates{(0.02,26.1962196)
(0.1,33.4715403)(0.2,65.018486)(0.3,1048.0275086)(0.4,1312.6205362)(0.5,1591.3042358)(0.6,3385.8635134)(0.7,4615.2120946)(0.8,5069.9626338)(0.9,4321.3572568)(1.0,1723.357002)(1.1,1742.5715472)(1.2,1842.5243542)(1.3,1852.0238146)(1.4,1856.2679218)(1.5,1857.9311232)(1.6,1860.3223286)(1.7,1862.0720272)(1.8,1861.651263)(1.9,1863.3709306)(2.0,1863.7405855)};\addlegendentry{T(16,8,8,8)+centralsymmetric};
\addplot+[gray,solid,mark=x,mark options=solid] coordinates{(0.02,24.7569456)
(0.1,30.6491081)(0.2,86.9111398)(0.3,722.00923)(0.4,1696.86438)(0.5,2130.3511884)(0.6,2375.3584152)(0.7,2445.0115704)(0.8,2429.9118778)(0.9,2205.6682374)(1.0,1815.6690068)(1.1,1981.2944672)(1.2,2058.8874516)(1.3,2062.3874676)(1.4,2067.6477228)(1.5,2076.575887)(1.6,2070.0514264)(1.7,2074.609896)(1.8,2073.919956)(1.9,2090.3609508)(2.0,2088.5489288)};\addlegendentry{$4D$-$FCC$(8)+randompairing};
\addplot+[gray,dashed,mark=x,mark options=solid] coordinates{(0.02,26.2481944)
(0.1,33.9481707)(0.2,130.0052244)(0.3,931.4986546)(0.4,1958.8061834)(0.5,2490.3664082)(0.6,2792.3560972)(0.7,2801.2953714)(0.8,2736.2221808)(0.9,2491.1922364)(1.0,2061.1440386)(1.1,2178.8502296)(1.2,2221.7215316)(1.3,2248.331533)(1.4,2266.437053)(1.5,2245.5660786)(1.6,2247.024566)(1.7,2273.5857894)(1.8,2274.994758)(1.9,2264.8219846)(2.0,2279.2066708)};\addlegendentry{T(16,8,8,8)+randompairing};
        \end{axis}
        \end{tikzpicture}
        \end{center}
		\caption{Packet delays in T(16,8,8,8) and $4D$-$FCC$(8) under several synthetic traffics.}
		\label{fig:simulation_latency_large}
\end{figure}
}

\newcommand{\SIMULATIONLATENCYSMALL}{%
\begin{figure}[th]
        \begin{center}
        \pgfplotsset{minor grid style={dashed,very thin, color=blue!15}}
        \pgfplotsset{major grid style={very thin, color=black!30}}
		\tikzsetnextfilename{simulation_latency_small}
        \begin{tikzpicture}
        \begin{axis}[
                xmin=0,xmax=1.2,
                ymin=0.0,
                ymax=900,
                ymajorgrids=true,
                yminorgrids=true,
                xmajorgrids=true,
                minor y tick num=4,
				xlabel={offered load ($\text{phits}/(\text{cycle}\cdot \text{node})$)},
				ylabel={average latency (cycles)},
                legend style={at={(1.00,1.01)},anchor=south east,font=\scriptsize},legend columns=2,legend cell align=left,
                ,
        ]

\addplot+[blue,solid,mark=o,mark options=solid] coordinates{(0.02,21.5774816)
(0.1,23.9249897)(0.2,27.4914136)(0.3,31.9050204)(0.4,37.593415)(0.5,44.9155356)(0.6,54.7456066)(0.7,68.5281032)(0.8,89.0501638)(0.9,123.1799568)(1.0,215.7171022)(1.1,700.7319732)(1.2,1073.6408054)(1.3,1242.4379742)(1.4,1314.851494)(1.5,1346.4029408)(1.6,1364.0162696)(1.7,1372.560283)(1.8,1377.8044736)(1.9,1383.087731)(2.0,1385.204699)};\addlegendentry{$4D$-$BCC$(4)+uniform};
\addplot+[blue,dashed,mark=o,mark options=solid] coordinates{(0.02,22.6875868)
(0.1,25.8302087)(0.2,30.8427018)(0.3,37.700298)(0.4,47.2377374)(0.5,61.326344)(0.6,84.032726)(0.7,126.4712426)(0.8,301.6992448)(0.9,1062.3279028)(1.0,1474.1893758)(1.1,1610.6538048)(1.2,1662.6673434)(1.3,1686.481315)(1.4,1700.4265128)(1.5,1704.9247876)(1.6,1709.357413)(1.7,1716.4125774)(1.8,1716.5056228)(1.9,1719.1422864)(2.0,1720.0922449)};\addlegendentry{T(8,8,8,4)+uniform};
\addplot+[red,solid,mark=*,mark options=solid] coordinates{(0.02,23.7448798)
(0.1,27.1330183)(0.2,32.582095)(0.3,39.871261)(0.4,49.8535398)(0.5,64.209015)(0.6,86.5426046)(0.7,126.1644628)(0.8,706.0936566)(0.9,3035.2370118)(1.0,3617.5567626)(2.0,4390.0720534)};\addlegendentry{$4D$-$BCC$(4)+antipodal};
\addplot+[red,dashed,mark=*,mark options=solid] coordinates{(0.02,30.8612434)
(0.1,40.8764609)(0.2,63.6321038)(0.3,114.361591)(0.4,339.0201228)(0.5,3413.7195316)(0.6,4604.8755058)(0.7,5034.2768088)(0.8,5254.6975698)(0.9,5430.1915368)(1.0,5495.0297956)(2.0,5590.3211002)};\addlegendentry{T(8,8,8,4)+antipodal};
\addplot+[green,solid,mark=square*,mark options=solid] coordinates{(0.02,22.0480438)
(0.1,24.7937484)(0.2,30.0792268)(0.3,39.8015214)(0.4,65.217491)(0.5,362.7460274)(0.6,519.1341134)(0.7,1175.6965824)(0.8,1408.8574264)(0.9,1424.2742902)(1.0,1513.5301984)(1.1,1849.0563554)(1.2,2188.4738036)(1.3,2368.7214642)(1.4,2476.789861)(1.5,2512.853304)(1.6,2548.9381278)(1.7,2558.7223948)(1.8,2523.4268682)(1.9,2572.1393842)(2.0,2629.3700767)};\addlegendentry{$4D$-$BCC$(4)+centralsymmetric};
\addplot+[green,dashed,mark=square*,mark options=solid] coordinates{(0.02,22.6556002)
(0.1,26.1107773)(0.2,33.4197112)(0.3,48.3228728)(0.4,97.3184418)(0.5,1022.8015966)(0.6,2659.0035678)(0.7,3270.708679)(0.8,3660.2968994)(0.9,3317.6301402)(1.0,1302.5580624)(1.1,1451.9389348)(1.2,1547.5897016)(1.3,1558.9322078)(1.4,1563.6021754)(1.5,1565.1324918)(1.6,1567.5265382)(1.7,1567.298253)(1.8,1569.1807778)(1.9,1568.585082)(2.0,1569.4918841)};\addlegendentry{T(8,8,8,4)+centralsymmetric};
\addplot+[gray,solid,mark=x,mark options=solid] coordinates{(0.02,21.6005)
(0.1,24.7669958)(0.2,35.0771944)(0.3,150.174222)(0.4,615.1375116)(0.5,1179.534606)(0.6,1590.0496448)(0.7,1764.5485274)(0.8,1830.7017594)(0.9,1767.7112702)(1.0,1585.1740316)(1.1,1804.707747)(1.2,1872.9682324)(1.3,1902.0305086)(1.4,1911.0360588)(1.5,1900.9746326)(1.6,1904.0150158)(1.7,1909.984297)(1.8,1918.5488496)(1.9,1923.7587374)(2.0,1955.92431)};\addlegendentry{$4D$-$BCC$(4)+randompairing};
\addplot+[gray,dashed,mark=x,mark options=solid] coordinates{(0.02,22.678409)
(0.1,27.0153672)(0.2,61.5662646)(0.3,450.20426)(0.4,1110.9884394)(0.5,1600.0571364)(0.6,1885.1909712)(0.7,1981.0640234)(0.8,2053.2654122)(0.9,1915.602667)(1.0,1778.1644096)(1.1,2009.1179034)(1.2,2096.1259484)(1.3,2067.638819)(1.4,2114.9647538)(1.5,2116.3519602)(1.6,2109.0820224)(1.7,2099.96533)(1.8,2132.2252448)(1.9,2138.9455082)(2.0,2142.9755528)};\addlegendentry{T(8,8,8,4)+randompairing};
        \end{axis}
        \end{tikzpicture}
        \end{center}
		\caption{Packet delays in T(8,8,8,4) and $4D$-$BCC$(4) under several synthetic traffics.}
		\label{fig:simulation_latency_small}
\end{figure}
}

\newcommand{\PEAKBARSLARGE}{%
\begin{figure}[t]
        \begin{center}
        \pgfplotsset{minor grid style={dashed,very thin, color=blue!15}}
        \pgfplotsset{major grid style={very thin, color=black!30}}
		\tikzsetnextfilename{peakbars_large}
        \begin{tikzpicture}
        \begin{axis}[
				ybar,
                height=120pt, width=240pt,
                symbolic x coords={uniform,antipodal,centralsymmetric,randompairing},
                ymin=0.0, ymax=0.8,
                ymajorgrids=true, yminorgrids=true, xmajorgrids=true,
                minor y tick num=4,
				xlabel={traffic pattern},
				ylabel={peak throughput},
                legend style={at={(0.95,0.95)},anchor=north east},
				legend cell align=left,
				xticklabel={\scriptsize\tick},
				every x tick label/.append style={anchor=base,yshift=-7,xtick=data},
        ]
\addlegendentry{$T(16,8,8,8)$}\addplot coordinates{(uniform,.482)(antipodal,.243)(centralsymmetric,.250)(randompairing,.236)};
\addlegendentry{$4D$-$FCC(8)$}\addplot coordinates{(uniform,.725)(antipodal,.394)(centralsymmetric,.308)(randompairing,.241)};
        \end{axis}
        \end{tikzpicture}
        \end{center}
	\caption{Throughput peak in $T(16,8,8,8)$ and $4D$-$FCC(8)$ under several synthetic traffics.}
	\label{fig:simulation_large_peaks}
\end{figure}
}

\newcommand{\PEAKBARSSMALL}{%
\begin{figure}[t]
        \begin{center}
        \pgfplotsset{minor grid style={dashed,very thin, color=blue!15}}
        \pgfplotsset{major grid style={very thin, color=black!30}}
		\tikzsetnextfilename{peakbars_small}
        \begin{tikzpicture}
        \begin{axis}[
				ybar,
                height=120pt, width=240pt,
                symbolic x coords={uniform,antipodal,centralsymmetric,randompairing},
                ymin=0.0, ymax=1.1,
                ymajorgrids=true, yminorgrids=true, xmajorgrids=true,
                minor y tick num=4,
				xlabel={traffic pattern},
				ylabel={peak throughput},
                legend style={at={(0.95,0.95)},anchor=north east},
				legend cell align=left,
				xticklabel={\scriptsize\tick},
				every x tick label/.append style={anchor=base,yshift=-7,xtick=data},
        ]
\addlegendentry{$T(8,8,8,4)$}\addplot coordinates{(uniform,.823)(antipodal,.399)(centralsymmetric,.400)(randompairing,.289)};
\addlegendentry{$4D$-$BCC(4)$}\addplot coordinates{(uniform,1.04)(antipodal,.700)(centralsymmetric,.580)(randompairing,.335)};
        \end{axis}
        \end{tikzpicture}
        \end{center}
	\caption{Throughput peak in $T(8,8,8,4)$ and $4D$-$BCC(4)$ under several synthetic traffics.}
	\label{fig:simulation_small_peaks}
\end{figure}
}

\newboolean{debug}
\setboolean{debug}{true}

\ifthenelse{\boolean{debug}}{
        \def\fixme#1{\typeout{FIXME in page \thepage : {#1}}\bgroup \color{red}{[{#1}]}\egroup}
        \def\note#1{\typeout{NOTE in page \thepage : {#1}}\bgroup \color{blue}{[{#1}]}\egroup}
}{
        \def\fixme#1{}{}
        \def\note#1{}{}
}

\renewcommand{\L}{\mathcal{L}}
\newcommand{\xor}{\mathrm{\scriptsize\bf xor}}

\newcommand{\Q}{\mathbb{Q}}

\newcommand{\Z}{\mathbb{Z}}

\newtheorem{proposition}{Proposition}
\newtheorem{definition}[proposition]{Definition}
\newtheorem{notation}[proposition]{Notation}
\newtheorem{theorem}[proposition]{Theorem}

\newtheorem{lemma}[proposition]{Lemma}
\newtheorem{remark}[proposition]{Remark}

\newtheorem{example}[proposition]{Example}
\catcode`\¿=\active\def¿{?`}

\begin{document}

\title{Symmetric Interconnection Networks from Cubic Crystal Lattices}
\author{Crist\'obal Camarero, Carmen Mart\'inez and Ram\'on Beivide
    \\Electronics and Computers Department. University of Cantabria}
\date{}
\maketitle

\begin{abstract}
Torus networks of moderate degree have been widely used in the supercomputer industry. Tori are superb when used for executing applications that require near-neighbor communications. Nevertheless, they are not so good when dealing with global communications. Hence, typical 3D implementations have evolved to 5D networks, among other reasons, to reduce network distances. Most of these big systems are mixed-radix tori which are not the best option for minimizing distances and efficiently using network resources. This paper is focused on improving the topological properties of these networks.

By using integral matrices to deal with Cayley graphs over Abelian groups, we have been able to propose and analyze a family of high-dimensional grid-based interconnection networks. As they are built over $n$-dimensional grids that induce a regular tiling of the space, these topologies have been denoted \textsl{lattice graphs}. We will focus on cubic crystal lattices for modeling symmetric 3D networks. Other higher dimensional networks can be composed over these graphs, as illustrated in this research. Easy network partitioning can also take advantage of this network composition operation.
Minimal routing algorithms are also provided for these new topologies. Finally, some practical issues such as implementability and preliminary performance evaluations have been addressed.
\end{abstract}


\section{Introduction}

Interconnection networks are critical subsystems in modern supercomputers. Currently, the top 5 supercomputers composed of Cray XK7, IBM BluGene/Q and K computers, use moderate degree networks. The Cray employs a 3D torus whereas BlueGene uses a 5D one, \cite{Cray,BGQnetwork}. The K computer employs small 3D meshes (that can also be seen as $4 \times 3$ tori) connected by a bigger 3D torus \cite{Tofu}. All these topologies are mixed-radix torus, as they have dimensions of different sizes. For example, a configuration for a Cray Jaguar can be $25\times 32\times 16$ and a BlueGene configuration $16 \times 16 \times 16 \times 12 \times 2$. The $88,128$-node K computer installed at Riken, is compatible with a $17 \times 18 \times 24$ torus connecting 3D meshes of 12 nodes. Mixed-radix tori are not edge-symmetric, which can lead to unbalanced use of their network links. However, these big systems are typically divided into smaller partitions which enables them to be used by multiple users. Hence, providing symmetry, at least, in typical network partitions is an advisable design goal.\par

Tori are not well suited to support global and remote communications. Their relatively long paths among nodes, especially their diameter and average distance, incur high latencies and reduced throughput. Thus, reducing  topological distances in the network should be pursued. The way to achieve network distance reductions is by changing the topology. Topological changes depend on the router degree. If the router degree must be kept within the current values, it would be interesting to preserve the good topological properties of tori such as grid locality, easy partitioning and simple routing. Hence, practicable topological changes should not be radical. A typical technique employed to this end has been twisting the wrap-around links of tori, \cite{IlliacIV,Sequin,Beivide,Martin}. Interestingly, this twisting also allows for edge-symmetric networks of sizes for which their corresponding tori are asymmetric, \cite{Camara,CamareroTC}. Twisting 2D tori is nearly as old as the history of supercomputers. The Illiac IV developed in 1971 already employed a twisted network. Many works dealing with twisted 2D tori have been published since then. However, when scaling dimensions, the problem of finding a good twisting scheme becomes harder. Very few solutions are known for 3D, with the one presented in \cite{Camara} being a practicable example. Exploring the effect of twists in higher dimensions remains, to our knowledge, an unexplored domain. If the router degree can be increased, a radically different solution for reducing network diameter can be used in high-degree hierarchical networks, \cite{Cascade}. These direct networks employing high-degree routers are beyond the scope of this paper.\par

It has been recognized for a long time that Cayley graphs are well suited to interconnection networks. Actually, the widely used rings and tori are Cayley graphs. Nowadays, rings are common in on-chip networks \cite{XeonRing} and, as stated previously, tori dominate high-end supercomputing. In \cite{Fiol}, Fiol introduced multidimensional circulant graphs as a new algebraic representation for Cayley graphs over Abelian groups. This representation has proved its suitability for studying and characterizing 2D grid-based networks in \cite{CamareroTC}. In this paper, \textsl{lattice networks} are introduced as multidimensional circulants with orthonormal adjacencies, that is, multidimensional grids plus additional wrap-around links which complete their regular adjacency. Therefore, this work is devoted to the study of high dimensional twisted tori topologies. Although special attention will be devoted to symmetric 3D networks, higher dimensional topologies which embed these symmetric 3D networks will also be considered. Specifically, the main contributions of this paper are:

\begin{itemize}
\item A characterization of 3D symmetric networks which correspond to cubic crystal lattices.
\item A general method for lifting crystal graphs which leads to higher dimensional lattice networks that embed crystal networks.
\item A minimal routing mechanism that performs over any lattice network.
\item A first approach to practical issues such as implementability and a preliminary performance evaluation of these networks, which includes both topological models and empirical simulations.
\end{itemize}

The remainder of this paper is organized as follows. Section \ref{sec:lattice} defines lattice graphs, introduces the concepts of graph lift and projection and provides some network examples. Section \ref{sec:cubic} focuses on 3D networks, describes symmetric cubic crystal graphs and performs a topological comparison of these networks with standard mixed-radix tori. Section \ref{sec:lifting} introduces two methods for scaling crystal networks to higher dimensions and presents some examples. Section \ref{sec:routing} presents minimal routing algorithms for lattice networks. Section \ref{sec:practical} discusses implementability and performance issues. Finally, Section \ref{sec:conclu} concludes the paper summarizing its main findings.

\section{Lattice Graphs}\label{sec:lattice}

In this section we introduce \textsl{lattice graphs} which will be used to model interconnection networks of any finite dimension. The lattice graph is not a new concept, in fact, it has many different uses. The most extensively used, which is the one used in this paper, is as a graph built over an $n$-dimensional grid which induces a regular tiling of the space. On the other hand, lattice graphs also appear in the literature under other names for example, tiling graphs. Moreover, in \cite{Fiol}, multidimensional circulants were defined as lattice graphs but for any set of adjacencies (not only the orthonormal adjacencies considered in this work), which \textsl{a priori} can seem to be a wider family of graphs. However, it can be seen that any multidimensional circulant can be transformed into a lattice graph. Hence, the study presented in this section is devoted, in fact, to the family of Cayley graphs over finite Abelian groups. Later on in this section, the concepts of \textsl{projection} and \textsl{lift} of a lattice graph will be stated. Projecting a lattice graph allows the study of the different lattice graphs of smaller dimensions which are embedded in it, while lifting a lattice graph will be used for increasing its dimension.\par

Lattice graphs are defined over the integer lattice $\mathbb{Z}^{n}$. Hence, their nodes are labelled by means of $n$-dimensional (column) integral vectors. A lattice graph can be intuitively seen as a multidimensional grid with additional wrap-around links completing the regular adjacency. Before proceeding with their formal definition, first we introduce some notation.

\begin{notation} The following notation will be used throughout the article:
\begin{itemize}
\item Lower case letters denote integers: $a$, $b$, $\ldots$
\item Bold font denotes integer column vectors: $\mathbf{v}$, $\mathbf{w}$, $\ldots$
\item Capitals correspond to integral matrices: $M$, $P$, $\ldots$
\item $\mathbf{e}_i$ denotes the vector with a $1$ in its $i$-th component and $0$ elsewhere.
\item $\mathcal{B}_{n} = \{\mathbf{e}_i \ | \ i=1, \dotsc, n\}$ denotes the $n$-dimensional orthonormal basis.
\end{itemize}
\end{notation}

To define the finite set of nodes of these graphs and their wrap-around links, a modulo function using a square integer matrix will be used. Hence, congruences modulo matrices are introduced in the next definition.

\begin{definition} \cite{Fiol} Let $M \in\mathbb{Z}^{n\times n}$
be a non-singular square matrix of dimension $n$. Two vectors $\mathbf{v}, \mathbf{w} \in \mathbb{Z}^{n}$ are congruent modulo $M$ if and only if we have
$\mathbf u = \begin{pmatrix} u_1\\ u_2\\ \vdots \\ u_n \\ \end{pmatrix} \in \mathbb{Z}^{n}$ such that: $$\mathbf{v}-\mathbf{w} = u_1 \mathbf{m}_1 + u_2 \mathbf{m}_2 + \dotsb +u_n\mathbf{m}_n = M\mathbf u$$
where $\mathbf{m}_{j}$ denotes the $j$-th column of $M$. We will denote this congruence as $\mathbf{v} \equiv \mathbf{w} \pmod{M}$.
\end{definition}

The set of nodes of a lattice graph will be the elements of the quotient group $$\mathbb{Z}^{n}/M\mathbb{Z}^{n} = \{\mathbf{v} \pmod{M} \ | \ \mathbf{v} \in \mathbb{Z}^{n}\}$$ generated by the equivalence relation induced by $M$. As was proved in \cite{Fiol}, $\mathbb{Z}^{n}/M\mathbb{Z}^{n}$
has $| \det(M) |$ elements. Now, we can proceed with a formal definition of a lattice graph.

\begin{definition} Given a square non-singular integral matrix $M \in \mathbb{Z}^{n \times n}$, we define the \textsl{lattice graph} generated by $M$ as $\mathcal{G}(M)$, where:
\begin{enumerate}
\item The vertex set is $\mathbb{Z}^{n}/M\mathbb{Z}^{n} = \{\mathbf{v} \pmod{M} \ | \ \mathbf{v} \in \mathbb{Z}^{n}\}$.
\item Two nodes $\mathbf{v}$ and $\mathbf{w}$ are adjacent if and only if $\mathbf{v}-\mathbf{w} \equiv \pm \mathbf e_i \pmod{M}$ for some $i\in\{1, \dotsc, n\}$.
\end{enumerate}
\end{definition}

From here onwards, all matrices will be considered to be non-singular, unless the contrary is stated. Note that, since $\mathbb{Z}^{n}/M\mathbb{Z}^{n}$ has $|\det (M)|$ elements, this will be the number of nodes of $\mathcal{G}(M)$. Moreover, since any vertex $\mathbf{v}$ is adjacent to $\mathbf{v} \pm \mathbf{e}_i \pmod{M}$, the lattice graph $\mathcal{G}(M)$ is regular of degree $2n$, that is, any node has $2n$ different neighbours. As stated in the following two paragraphs, tori are lattice graphs.

\begin{definition} The $n$-dimensional torus graph of sides $a_1,\dotsc,a_n$,
denoted by $T(a_1,\dotsc,a_n)$ is defined as a graph with vertices $\mathbf{x} \in \Z ^n$ such that
$0\leq x_i<a_i$. Two vertices $\mathbf x$ and $\mathbf y$ are adjacent if and only if
they differ in exactly one coordinate, let us say $i$, for which $x_i\equiv y_i \pm 1 \pmod{a_i}$.
\end{definition}

\begin{theorem}\label{th:torus} The torus graph $T(a_1,\dotsc,a_n)$ is isomorphic
to the lattice graph $\mathcal G(\mathrm{diag}(a_1,\dotsc,a_n))$, where $\mathrm{diag}(a_1,\dotsc,a_n)$ denotes the square diagonal
matrix with diagonal equal to $a_1,\dotsc,a_n$.
\end{theorem}

\begin{proof}
Clearly the vertex space of both graphs is the same.
In the following we check that adjacencies are preserved.
If $\mathbf x$ is connected to $\mathbf y$ then it holds that
	$\mathbf y-\mathbf x
	=(y_i-x_i)\mathbf e_i$.
Then for some integer $k$,
	$\mathbf y-\mathbf x
	=(\pm1 +ka_i)\mathbf e_i
	=\pm \mathbf e_i +ka_i\mathbf e_i
	=\pm \mathbf e_i +\mathrm{diag}(a_1,\dotsc,a_n)k\mathbf e_i
	$.
Hence $\mathbf y-\mathbf x\equiv \pm \mathbf e_i \pmod{\mathrm{diag}(a_1,\dotsc,a_n)}$.
\end{proof}

Next we recall some known results from \cite{Fiol} about right-equivalent matrices.

\begin{definition} \label{def:rightequivalent} $M_1$ is \textsl{right equivalent} to $M_2$, which is denoted by $M_1 \cong M_2$, if and only if there exists a unitary matrix $P \in \mathbb{Z}^{n \times n}$ such that
$M_1 = M_2 P$.
\end{definition}

As was proved in \cite{Fiol}, if $M_1 \cong M_2$ then the graphs $\mathcal{G}(M_1)$ and $\mathcal{G}(M_2)$ are isomorphic. As a consequence, performing Gaussian elimination by columns in the generating matrix gives isomorphic graphs. After one phase of Gaussian elimination we obtain:
$$M  \cong \begin{pmatrix}B& \mathbf{c} \\0&a\end{pmatrix}$$
where $B \in \mathbb{Z}^{n-1 \times n-1}$ is a matrix of smaller dimension, $\mathbf{c} \in \mathbb{Z}^{n-1}$ is a column vector and $a$ is a positive integer. As a consequence, we obtain that $|\det(M)|=|\det(B)|a$, that is, the order of $\mathcal{G}(M)$ can be expressed in terms of $\mathcal{G}(B)$ and the integer $a$. Moreover, the lattice graph $\mathcal{G}(B)$ is isomorphic to the subgraph of $\mathcal{G}(M)$ generated by $\{\pm \mathbf{e}_1,\pm \mathbf{e}_2,\dots,\pm \mathbf{e}_{n-1}\}$, which allows us to state the following definition.

\begin{definition}\label{def:projection} Let $M \in \mathbb{Z}^{n \times n}$ be non-singular and $\mathcal{G}(M)$ be its lattice graph. Let us consider $M \cong \begin{pmatrix}B& \mathbf{c} \\0&a\end{pmatrix}$
such that $a$ is a positive integer. Then, we will say that $a$ is the \textsl{side} of $\mathcal{G}(M)$ and $\mathcal{G}(B)$ its \textsl{projection} over $\mathbf{e}_{n}$. Moreover, we will call $\mathcal{G}(M)$ a \textsl{lift} of $\mathcal{G}(B)$.
\end{definition}

In particular, any lattice graph can be considered to be generated by its unique Hermite matrix, which may be convenient as Examples \ref{eje:1} and \ref{eje:2} attempt to demonstrate. Before stating the examples, we recall the Hermite normal form of a matrix.

\begin{definition}\label{def:hermite}
A matrix $H$ is said to be in \emph{Hermite normal form} if it is upper
triangular, has positive diagonal and each $H_{i,j}$ with $j>i$ lies in a
complete set of residues modulo $H_{i,i}$.
\end{definition}

Definitions \ref{def:projection} and \ref{def:hermite} allow us to consider a helpful graphical visualization of any lattice graph which will also be used for routing in Section \ref{sec:routing}. First, lattice graphs and their subgraphs can be seen as $n$-dimensional spaces whose dimensions are sized by the elements in the principal diagonal of $M$. Each column vector in $M$ represents a graph dimension, signaling the point in the space at which a new copy of the tile induced by $M$ is located; this is important as column vectors dictate the pattern of the wrap-around connections of each dimension.\par

Moreover, from the cardinal equality $|\mathcal{G}(M)|=|\mathcal{G}(B)|a$, the lattice graph $\mathcal{G}(M)$ can be seen as composed of $a$ disjoint copies of its projection $\mathcal{G}(B)$. One or several parallel cycles connect these disjoint copies completing the adjacency pattern. The length of these cycles can be computed as $ord(\mathbf{e}_{n})$, which is the order of the element $\mathbf{e}_{n}$ in the group $\mathbb{Z}^{n}/M\mathbb{Z}^{n}$. According to \cite{Fiol}, the order of any element $\mathbf{x}$ can be computed as $$\frac{\det(M)}{\gcd(\det(M),\gcd(\det(M) M^{-1}\mathbf{x}))}.$$
Note that the second $\gcd$ (greatest common divisor) in the fraction corresponds to the $\gcd$ of the elements of a vector. The number of vertices of each cycle lying in each copy of $\mathcal{G}(B)$ can be calculated as the length of the cycle over the side of the graph, that is $\frac{ord(\mathbf{e}_{n})}{a}.$

\begin{example}\label{eje:1} Let us consider the \textsl{rectangular twisted torus} of size $2a \times a$ and twist $a$, denoted as $RTT(a)$ in \cite{Camara}.
This graph can be seen to be generated by the matrix $H=\begin{pmatrix}2a&a\\0&a\end{pmatrix}$. Using $H$, the graph can be seen as a grid of $2a\times a$ ($h_{1,1}\times h_{2,2}$). Wrap-around links in $\mathbf{e}_1$ (first) dimension conserve their horizontality since $h_{2,1}=0$; wrap-around links in $\mathbf{e}_2$ (second) dimension do not conserve their verticality but suffer a twist of $a$ columns since $h_{1,2}=a$. According to Definition \ref{def:projection}, the projection over $\mathbf{e_2}$ of $RTT(a)$ is a cycle of $2a$ nodes each. As the side of $RTT(a)$ is $a$, it will have $a$ disjoint cycles of $2a$ nodes. As $ord(\mathbf{e}_{2})$ (the element representing a jump in $\mathbf{e}_2$ dimension) is $2a$, the graph will have $a$ parallel cycles of length $2a$ in that dimension. Each of these $a$ cycles contains two vertices of each projection. A graphical representation of $RTT(4)$ can be seen in Figure \ref{fig:RTTcycles}.\end{example}

\RTTCYCLES

\begin{example}\label{eje:2} Let us now consider the lattice graph $\mathcal{G}(M)$ with $M=\begin{pmatrix}4&0&0\\0&4&2\\0&0&4\end{pmatrix}$.
Note that $M$ is in Hermite form. $\mathcal{G}(M)$ can be seen as a $4\times 4\times 4$ cubic grid. Three sets of wrap-around links, each one connecting opposite faces, have to be added to the grid-based cube. Wrap-around links in $\mathbf{e}_1$ always remain horizontal by construction, as imposed by the $n-1$ zeros in the first column vector of any Hermite matrix. Wrap-around links in the $\mathbf{e}_2$ dimension remain vertical in this graph because $m_{1,2}=0$ but, in general, they can undergo only a twist over the $\mathbf{e}_1$ dimension of $m_{1,2}$ units. Finally, wrap-around links in the $\mathbf{e}_3$ dimension can undergo twists over both $\mathbf{e}_1$ and $\mathbf{e}_2$ dimensions. In the graph of this example, no twist is applied in $\mathbf{e}_3$ over $\mathbf{e}_1$ because $m_{1,3}=0$ and a twist of 2 units is applied over the $\mathbf{e}_2$ dimension as $m_{2,3}=2$. As can be seen in Figure \ref{fig:joiningcycle}, the projection of $\mathcal{G}(M)$ is $\mathcal{G}\begin{pmatrix}4&0\\0&4\end{pmatrix}$, a 2D torus $T(4, 4)$. Thus, the graph is composed of 4 disjoint copies of its projection, each of them connected by a cycle of length 8, as represented in the figure. Note that for every vertex in the graph there will be a similar cycle with the same pattern as the one represented in the figure. The cycle intersects in two vertices with each copy of the projection. For the sake of the clarity, only one cycle between copies of $\mathcal{G}\begin{pmatrix}4&0\\0&4\end{pmatrix}$ has been represented.\end{example}

\JOININGCYCLETRIDIMENSIONALSHADOWAXIS

Note that we can project over any $\mathbf{e}_i$, simply by swapping rows $i$ and $n$ (which gives an automorphic graph) and then, project over $\mathbf{e}_n$. Moreover, as we will see later, symmetries will make irrelevant over which dimension we project, so we will consider $\mathbf{e}_n$ by default. The resulting projection can again be projected over another vector, which results in a projection over a plane of the lattice graph.
Clearly, projecting over a pair of vectors $\{\mathbf e_i , \mathbf e_j \}$ can be done in any order, since projecting first over $\mathbf e_i$ and then over $\mathbf e_j$ results in the same graph as projecting first over $\mathbf e_j$ and then over $\mathbf e_i$. Following the same idea, we can project over several dimensions iteratively.
Therefore, we will call the result of projecting iteratively over the vectors
in the set $\{\mathbf e_{i_1}, \dotsc , \mathbf e_{i_r} \}$ the
\textsl{projection} of $\mathcal{G}(M)$ \textsl{over the set}.  In this case we
will call it a $r$-dimensional projection which turns into a lattice graph
generated by a $(n-r) \times (n-r)$ matrix.

\section{Cubic Crystal Graphs}\label{sec:cubic}

Symmetry is a desirable property for any network as it impacts on performance and routing efficiency. Many interconnection networks have been based
on vertex-symmetric graphs, but less attention has been devoted to edge-symmetric networks. Square and cubic tori have been the networks
of choice for many designs as they are symmetric (vertex and edge symmetric). For this reason, symmetric lattice graphs will be considered in this section. Hence, we next introduce the concept of a symmetric graph.\par

A graph $G=(V,E)$ is \emph{vertex-symmetric} (or \emph{vertex-transitive})
if for each pair of vertices $(x,y)\in V$, there is an automorphism $\phi$ of $G$ such that
$\phi(x)=y$. Also, $G$ is \emph{edge-symmetric} (or \emph{edge-transitive})
if for each pair of edges $(\{x_1,x_2\},\{y_1,y_2\})\in E$, there is an
automorphism $\phi$ of $G$ such that $\phi(\{x_1,x_2\})=\{\phi(x_1),\phi(x_2)\}=\{y_1,y_2\}$.
Finally, $G$ is said to be \emph{symmetric} when it is both vertex-symmetric and edge-symmetric. Since every Cayley graph is vertex-symmetric \cite{Krishna}, we will focus on edge-symmetry.
As was shown in \cite{IWONT10}, the consideration of non-linear automorphisms in the edge-symmetry characterization leads to marginal families of
graphs which do no exemplify the general behaviour. Hence, in this paper we will refer only to automorphisms which are
linear applications. Therefore, in an abuse of notation, symmetric graphs will refer to those in which there exist a linear automorphism fulfilling the previous definition.

\begin{theorem}\label{theo:isoprojection} The projections of a symmetric lattice graph are all isomorphic.
\end{theorem}

\begin{proof}
Let us denote $proj_i(\mathcal G(M))$ to be the projection of $\mathcal G(M)$ over $\mathbf{e}_i$. We know $proj_i(\mathcal G(M))$ is isomorphic to the subgraph of $\mathcal G(M)$
generated by $\mathcal B_n\setminus \{\mathbf e_i\}$.
As $\mathcal G(M)$ is symmetric we know $\phi \in Aut(\mathcal G(M))$ such that
$\phi(\mathbf e_i)=\pm \mathbf e_j$.
As $\mathbf e_i$ is the only generator not in $proj_i(\mathcal G(M))$,
$\mathbf e_j$ is the only generator not in $\phi(proj_i(\mathcal G(M)))$.
Hence, as $\phi$ is an automorphism, we deduce that $proj_i(\mathcal G(M))\cong proj_j(\mathcal G(M))$.
\end{proof}

Now, we concentrate on 3D symmetric graphs. In Appendix I it is proved that the only symmetric 3D lattice graphs are the ones given by the matrices described in the next result.

\begin{theorem}\label{theo:simetricos} Let $M \in \mathbb{Z}^{3 \times 3}$. Then, the lattice graph $\mathcal{G}(M)$ is symmetric if and only if it is isomorphic to $\mathcal{G}(M')$
where: $$M' \in \left \{ \begin{pmatrix}a&c&b\\b&a&c\\c&b&a\end{pmatrix}, \begin{pmatrix}a&b&c\\a&c&-b-c\\a&-b-c&b\end{pmatrix} \right \}.$$
\end{theorem}

The previous characterization gives us a broad family of symmetric graphs. However, there are matrices belonging to the first case that deserve special attention, such as the ones that generate the cubic crystal lattices \cite{Janssen}, which are:

\begin{itemize}
\item \textbf{Primitive Cubic Lattice:} $\begin{pmatrix}a&0&0\\0&a&0\\0&0&a\end{pmatrix}$.
\item \textbf{Face-centered Cubic Lattice:} $\begin{pmatrix}a&a&0\\a&0&a\\0&a&a\end{pmatrix}$.
\item \textbf{Body-centered Cubic Lattice:} $\begin{pmatrix}-a&a&a\\a&-a&a\\a&a&-a\end{pmatrix}$.
\end{itemize}

In the following subsections we will consider the lattice graphs defined by cubic crystal lattices, their isomorphisms with previously studied network topologies and a comparison among them in terms of their distance properties.

\CUBICGRAPHS

\subsection{Primitive Cubic lattice graph}

We define the \textbf{Primitive Cubic Lattice Graph} $PC(a)$ as the lattice graph generated by the matrix associated with the primitive cubic lattice, that is:
$$\begin{pmatrix}a&0&0\\0&a&0\\0&0&a\end{pmatrix}.$$

Clearly, the order of the graph is $a^3$, which is the determinant of the diagonal matrix. According to Theorem \ref{th:torus}, $PC(a)$ is isomorphic to the 3D torus of side $a$, or equivalently, the $a$-ary 3-cube.

\begin{lemma} The projection of $PC(a)$ is the 2D torus graph of side $a$ or $\mathcal{G}(\begin{pmatrix}a&0\\0&a\end{pmatrix})$.
\end{lemma}

\subsection{Face-centered Cubic lattice graph}

The \textbf{Face-centered Cubic lattice graph} $FCC(a)$ of side $a$ can be defined as
the lattice graph generated by the matrix associated with the face-centered cubic crystal lattice, that is:
$$\begin{pmatrix}a&a&0\\a&0&a\\0&a&a\end{pmatrix} \cong \begin{pmatrix}2a&a&a\\0&a&0\\0&0&a\end{pmatrix}.$$

The order of the graph is $|\det (M)| = 2|a|^3$.

\begin{lemma}\label{lem:FCC} The projection of $FCC(a)$ is the rectangular twisted torus graph of side $a$, $RTT(a)$. \end{lemma}

\begin{proof} After performing Gaussian elimination, on the right of the previous expression we obtained the Hermite form of the matrix.
It is easy to see that its projection is generated by $\begin{pmatrix}2a&a\\0&a\end{pmatrix}$. As we have seen before and was proved in \cite{CamareroTC}, this graph is isomorphic to the rectangular twisted torus $RTT(a)$ of side $a$ or the \textsl{Gaussian graph} generated by $a+ai$ \cite{IEEE_TC}.
\end{proof}

A $FCC(a)$ is isomorphic to the \textsl{prismatic doubly twisted torus} of side $a$ ($PDTT(a)$), introduced in \cite{Camara}, as the next proposition proves.

\begin{proposition} $FCC(a)$ is isomorphic to the prismatic doubly twisted torus of side $a$, $PDTT(a)$.
\end{proposition}

\begin{proof}
The $PDTT(a)$ was defined in \cite{Camara} as a graph in which the connectivity of each plane
is a $RTT(a)$, hence the isomorphism is immediate once we have proved that all the projections
of $FCC(a)$ are isomorphic to $RTT(a)$. Note that this fact can be inferred from
Lemma \ref{lem:FCC} and Theorem \ref{theo:isoprojection}.
\end{proof}

\subsection{Body-centered Cubic lattice graph}

The \textbf{Body-centered Cubic lattice graph} $BCC(a)$ of side $a$ can be defined as
the lattice graph generated by the matrix:
$$\begin{pmatrix}-a&a&a\\a&-a&a\\a&a&-a\end{pmatrix} \cong \begin{pmatrix}2a&0&a\\0&2a&a\\0&0&a\end{pmatrix}.$$
The order of the graph is $4a^3$. As far as we know, this graph has not previously been considered for interconnection networks.
However, as we will see later, the graph not only meets the symmetry requirements but also has a good order/diameter correspondence. Moreover, it embeds 2D symmetric tori as is proved in:

\begin{lemma} The projection of $BCC(a)$ is the 2D torus graph $T(2a,2a)$
\end{lemma}

\begin{proof} It can be verified that after performing Gaussian elimination over the original matrix, it is easy to see that its projection is generated by $\begin{pmatrix}2a&0\\0&2a\end{pmatrix}$ which is the 2D torus of side $2a$.
\end{proof}

\subsection{Cubic crystal lattice graph comparison} \label{sub:crystal_comp}

In previous subsections, three different 3D symmetric topologies based on cubic crystal lattices have been introduced. As we have seen, two of them --the 3D torus or $PC$ and the PDTT or $FCC$--, were previously known, and the last one, that is the $BCC$, is a new proposal introduced in this paper. In this subsection, our aim is to consider their distance properties and to perform a first comparison in terms of diameter, average distance and projections.\par

First of all, we would like to highlight that a cubic crystal lattice graph exists for any order that is a power of two. This is important because we can gracefully upgrade a network in three steps while conserving symmetry. If $t$ is a positive integer, then:

\begin{itemize}
\item There exists a primitive cubic lattice graph with $2^{3t}$ nodes.
\item There exists a face-centered cubic lattice graph with $2^{3t+1}$ nodes.
\item There exists a body-centered cubic lattice graph with $2^{3t+2}$ nodes.
\end{itemize}

Although this fact provides practical versatility, it complicates the comparison among networks. The exact expressions for average distance of the three crystals are given next:

$PC(a)$ has average distance:
    $$\bar k=\begin{cases}
        \frac{3a^4}{4(a^3-1)} & \text{if }2|a\\
        \frac{3a^4-3a^2}{4(a^3-1)} & \text{if }2\!\!\not|a\\
    \end{cases}$$

$FCC(a)$ has average distance:
    $$\bar k=\begin{cases}
        \frac{7a^4-2a^2}{4(2a^3-1)} & \text{if }2|a\\
        \frac{7a^4-2a^2-1}{4(2a^3-1)} & \text{if }2\!\!\not|a\\
    \end{cases}$$

$BCC(a)$ has average distance:
    $$\bar k=\begin{cases}
        \frac{35a^4-8a^2}{8(4a^3-1)} & \text{if }2|a\\
        \frac{35a^4-14a^2+30}{8(4a^3-1)} & \text{if }2\!\!\not|a\\
    \end{cases}$$

These expressions have been calculated under the assumption that the average distance fulfills a polynomial expression, which
is a reasonable hypothesis. Moreover, these values have been computationally checked for orders up to $40,000$. In Table \ref{table:3D} the distance properties for the three graphs are summarized. For an easier comparison, note that average distance values are given as approximations. Mixed-radix torus graphs which have the same number of nodes of the $FCC$ and $BCC$ crystals have been also added in the table. Clearly, the crystals have better distance properties than their corresponding torus networks. Moreover, $BCC$ is more dense than the other two cubic crystals since, for the same diameter, it attains a greater number of nodes. Finally, as we have seen in previous subsections, while $FCC$ has the twisted torus as its projection, both $PC$ and $BCC$ are lifts of a 2D symmetric torus graph.

\begin{table}
    \begin{center}
    \def\arraystretch{1.4}
    \begin{tabular}{|llll|}
    \hline
    Topology        &Nodes        &Diameter            &Average Distance\\
    \hline
    $PC(a)$               &$a^3$        &$3\left\lfloor\frac{a}{2}\right\rfloor$        &$\approx\frac{3}{4}a=0.75a$\\
    $T(2a,a,a)$           &$2a^3$        &$a+2\left\lfloor\frac{a}{2}\right\rfloor$        &$\approx a$\\
    $FCC(a)$        &$2a^3$        &$\left\lfloor\frac{3}{2}a\right\rfloor$        &$\approx\frac{7}{8}a=0.875a$\\
    $T(2a,2a,a)$       &$4a^3$        &$\left\lfloor\frac{5}{2}a\right\rfloor$        &$\approx \frac{5}{4}a=1.25a$\\
    $BCC(a)$        &$4a^3$        &$\left\lfloor\frac{3}{2}a\right\rfloor$        &$\approx \frac{35}{32}a=1.09375a$\\
    \hline
    \end{tabular}
    \end{center} \caption{Distance properties of cubic crystal lattice graphs}\label{table:3D}
\end{table}

Having considered distance-related parameters for comparing crystals, let us also take into account other topological parameters to complete the study. In networking literature, the \emph{bisection bandwidth} ($BB$) is used to obtain an upper bound for the network load under uniform random traffic. However, it was shown in \cite{Camara} that in rectangular twisted tori some minimal routes
between pairs of vertices in opposite network partitions could traverse the bisection twice. Hence, this work proved that $BB$ is not a tight bound for network throughput in twisted topologies. Indeed, the same happens with any non-torus lattice graph.\par

There is another way to accurately bound network throughput under uniform traffic under ideal conditions. Throughput is inversely proportional to average distance in symmetric networks. As, under uniform traffic at rate $l$, $l$ phits are injected into each node each cycle, we have a total of $lN\bar k$ links being used each cycle. As a link can only transfer 2 phits (one in each way) each cycle, we have $lN\bar k\leq 2|E|=\Delta N$, where $\Delta$ denotes the graph degree and $N$ and $E$ denote the order and the edge set respectively. Thus, network throughput is bounded by $\dfrac{\Delta}{\bar k}$; for lattice graphs, $\Delta=2n$ where $n$ is the number of dimensions. Hence, $FCC(a)$ maximum throughput will be bounded by $\dfrac{48}{7a}$ and $BCC(a)$ by $\dfrac{192}{35a}$. Nevertheless, the previous count cannot be applied to edge-asymmetric networks such as mixed-radix tori. In that case, it can be seen that throughput is inversely proportional to the maximum average distance per dimension, namely $\dfrac{\Delta}{n \bar k_{max}}$, as inferred from \cite{Camara}. Network throughput for both $T(2a,a,a)$ and $T(2a,2a,a)$ is bounded by $\dfrac{12}{3a} = \dfrac{4}{a}$ as $\bar k_{max} \approx \dfrac{a}{2}$, given that their longest dimensions are $2a$-node rings. This leads to an improvement in maximum throughput under uniform traffic of 71\% when comparing $FCC(a)$ to $T(2a,a,a)$ and 37\% for $BCC(a)$ versus $T(2a,2a,a)$.\par

Being symmetric has more positive impact when the number of nodes is $2a^3$. In $T(2a,a,a)$, when the links in the longest dimension are fully utilized, links in the other two shortest dimensions are used at 50\%. This is because, on average, the length of the paths in the longest dimension doubles the length of the shortest ones. When the number of nodes is $4a^3$, $T(2a,2a,a)$ uses its resources better as only links in one dimension operate at half rate.

\section{Higher Dimensions: Lifts and Hybrid Graphs}\label{sec:lifting}

In the previous section we characterized 3D symmetric topologies and detailed the special case of the cubic crystal graphs.
Symmetry could help when the application runs on the whole network. However, in big systems the user typically only has a partition of the complete machine assigned.
Therefore, looking for symmetry in higher dimensions cannot be prioritized. Nevertheless, reducing the distance properties of
the whole network would be still beneficial since applications and system software sometimes run over the
entire network. Consequently, what we look for are higher dimensional networks embedding the
previous crystal cubic lattice graphs.\par

In the next subsections we explore two different methods for upgrading the previous cubic crystal lattice graphs. In the first subsection, we consider the lifting of crystal graphs, which results in 4D topologies. Whenever possible, the lift is done in such a way that the resulting eight-degree topology preserves symmetry. To end this subsection, we will introduce a tree that represents the process of network upgrading, preserving symmetry. In the second subsection we present \textsl{common lifts} of lattice graphs. The ultimate aim of this new method is to build new lattice graphs which embed other lattice graphs, while minimizing the necessary network degree to obtain them. The resulting graphs have been denoted as \textsl{hybrid graphs} since several lattice graphs of different nature (symmetric or non-symmetric) and degrees are embedded on them.

\subsection{Symmetric Lifts of Cubic Crystal Graphs}

First, we consider the $PC$. There is a straightforward way of lifting a $PC(a)$ to 4D, which is the Cartesian product of the $PC$ by one cycle of length $a$, thus obtaining the generator matrix:

$$\begin{pmatrix}a&0&0&0\\0&a&0&0\\0&0&a&0\\0&0&0&a\end{pmatrix}.$$

The 4D torus generated by the previous matrix is completely symmetric. However, the lifting technique can be used to embed the completely symmetric 3D torus in a different lattice graph. We will denote the \textsl{body centered hypercube lattice graph} as $4D$-$BCC$, that is, the lattice graph generated by matrix:

$$\begin{pmatrix}2a&0&0&a\\0&2a&0&a\\0&0&2a&a\\0&0&0&a\end{pmatrix}.$$

\begin{proposition}\label{prop:4BCC} $4D$-$BCC(a)$ is a symmetric lattice graph of side $a$ and projection $PC(2a)$.
\end{proposition}

\begin{proof}
Let $\phi$ be defined by $\phi(\mathbf e_i)=\mathbf e_{i+1\pmod{n}}$.
$\phi$ has an associated matrix $P=\begin{pmatrix}0&0&0&1\\1&0&0&0\\0&1&0&0\\0&0&1&0\end{pmatrix}$.
As $Q=M^{-1}PM=\begin{pmatrix}0&0&-1&0\\1&0&-1&0\\0&1&-1&0\\0&0&2&1\end{pmatrix}$
is an integer matrix we conclude that $\phi$ is an automorphism of $4D$-$BCC$ \cite{IWONT10}.
In the group generated by $\phi$ there are enough automorphisms to provide the edge-symmetry.
It should be noted that the projection is straightforward as the 
matrix is triangular superior.
\end{proof}

Now, if we want to lift the $FCC$, there are two ways of doing so which make the lifted graph symmetric. The first one will be denoted as $4D$-$FCC$ (\textsl{4-dimensional face-centered cubic lattice graph}), that is, the lattice graph generated by matrix:

$$\begin{pmatrix}2a&a&a&a\\0&a&0&0\\0&0&a&0\\0&0&0&a\end{pmatrix}.$$

\begin{proposition}
$4D$-$FCC(a)$ is a symmetric lattice graph of side $a$ whose projection is a $FCC(a)$.
\end{proposition}

\begin{proof}
Exactly like the proof of Proposition \ref{prop:4BCC}, the matrix $Q=M^{-1}PM$
is different but still with integer entries.
\end{proof}

The second way to lift $FCC$ is introduced below.

\begin{proposition} The lattice graph generated by the matrix $\begin{pmatrix}
a&-a&-a&-a\\
a&a&-a&a\\
a&a&a&-a\\
a&-a&a&a\\
\end{pmatrix}$ is a symmetric lifting of the $FCC(2a)$.
\end{proposition}

\begin{proof} The isomorphism is guaranteed since
$$
\begin{pmatrix}
a&-a&-a&-a\\
a&a&-a&a\\
a&a&a&-a\\
a&-a&a&a\\
\end{pmatrix}
\cong
\begin{pmatrix}
2a&-2a&0&-a\\
0&2a&-2a&a\\
2a&0&2a&-a\\
0&0&0&a\\
\end{pmatrix}$$

For symmetry, the procedure described in the proof of Proposition \ref{prop:4BCC} is repeated.
\end{proof}

This second lifting relates the graphs obtained to the family of Lipschitz graphs and quaternion algebras, introduced in \cite{IEEETITQuat},
for obtaining perfect codes over 4D spaces. This graph will be denoted as $Lip(a)$.

Finally, there are several ways of lifting the $BCC$, although none of them preserves the symmetry as proved in the next theorem.

\begin{theorem}
Any lift of $BCC$ yields a non-edge-symmetric graph.
\end{theorem}

\begin{proof} Let $M=\begin{pmatrix}2a&0&a\\0&2a&a\\0&0&a\end{pmatrix}$, $BCC(a)\simeq\mathcal G(M)$.
Assume that exists a symmetric lift $\mathcal G(L)$ of $BCC(a)$.
	$$L=\begin{pmatrix}2a&0&a&x\\0&2a&a&y\\0&0&a&z\\0&0&0&t\end{pmatrix}.$$
In Hermite form we have $0\leq x,y<2a$ and $0\leq z<a$. For symmetry, the $\gcd$ of every row
must be the same (map $\mathbf e_i$ into $\mathbf e_n$ and Gauss-reduce), hence $t$ divides all the other entries of $L$ and without loss of generality
we assume $t=1$.
By \cite{IWONT10} we know that automorphisms are matrices $P$ satisfying the condition that
$M^{-1}PM$ is an integer matrix where $P$ is unitary and only has $\pm1$ entries.
Both, the sets of these matrices which would give edge-transitivity,
and the possible lifts, are finite. Hence
we can run a computation which gives the negative result.\end{proof}

\arbolsimetricos

As we have concluded before, there is no decisive interest in obtaining a symmetric graph in 4D such that its 3D partitions remain themselves symmetric. Therefore, we could explore which of the lattice graphs whose projection is a $BCC$ would be the most interesting.\par

Figure \ref{fig:arbol} summarizes how the previous constructions can be generalized to any number of dimensions. The procedure is represented in a tree. In this tree, nodes are the matrices of the lattice graphs. Note that, for an easier visualization, matrices have been normalized by dividing them by $a$. Hence, each child is a lift of its parent. Moreover, we have restricted lifts to those whose side is greater or equal to the half of the side of its projection, otherwise many more graphs would appear.\par

The root of the tree is the matrix associated with a cycle. The lifts of the cycle conserving symmetry, and fulfilling the restrictions mentioned above, are the torus and the twisted torus introduced in Section \ref{sec:lattice}. Then, as we have seen in Section \ref{sec:cubic}, the cubic crystal lattice graphs are lifts of these two. The two branches show that only two families are obtained. The left branch consists of the infinite family of symmetric tori or $n$-dimensional $PCs$;
and each $nD$-$PC$ has a $nD$-$BCC$ sibling which is a leaf, without any further symmetric lift. The right branch is the family of the $n$-dimensional $FCC$s;
the $nD$-$FCC$ always has the $(n+1)D$-$FCC$ as a symmetric lift. Moreover, there are some dimensions (4 and 6 in the figure) in which a different lift exists. Interestingly, two non right-equivalent matrices generate the same graph (denoted with $\simeq$). The two branches in the tree are really different and, as we show next, they can be used to obtain new hybrid lattice graphs.

\subsection{Hybrid Graphs: Common Lift of Crystal Graphs}

In this subsection a different approach for embedding crystal graphs is considered. Given a number of crystal graphs, the idea is to generate a lattice graph which has them as its projections. Let us introduce this concept in the next definition.

\begin{definition} The lattice graph $\mathcal{G}(M)$ is a \textsl{common lift} of $\mathcal{G}(M_1)$ and $\mathcal{G}(M_2)$ if both can be obtained as projections of $\mathcal{G}(M)$.
\end{definition}

\begin{remark} There are several ways of obtaining different common lifts of two given lattice graphs. A straightforward one is to consider the lattice graph $\mathcal{G}(M_1 \oplus M_2)$ generated by the direct sum of the matrices. As we state next, this option leads to the Cartesian product of the two given lattice graphs.
\end{remark}

\begin{lemma} $\mathcal{G}(M_1 \oplus M_2)$ is a common lift of $\mathcal{G}(M_1)$ and $\mathcal{G}(M_2)$ and $\mathcal{G}(M_1 \oplus M_2) \cong \mathcal{G}(M_1) \times \mathcal{G}(M_2)$, which denotes the Cartesian product of $\mathcal{G}(M_1)$ and $\mathcal{G}(M_2)$.
\end{lemma}

As we will see next, there exist other common lifts which obtain $\mathcal{G}(M_1)$ and $\mathcal{G}(M_2)$ as projections but generating a lattice graph of smaller dimension. Note that this would be beneficial for cost aspects, such as minimizing the degree of the network routers, and to provide a good relation between the size of the graph and its projections.\par

\begin{theorem}\label{theo:minimumlift} Given two lattice graphs $\mathcal{G}(M_1)$ and $\mathcal{G}(M_2)$, we consider the lattice graph $\mathcal{G}(M_1 \boxplus M_2)$ which is obtained as follows:
Let $M_1\cong H_1$ and $M_2\cong H_2$ with $H_1$ and $H_2$ in Hermite normal form.
Let $C$ be the submatrix with the first common columns of $H_1$ and $H_2$.
Then $H_1=\begin{pmatrix}C&R_A\\ 0& A\end{pmatrix}$
and $H_2=\begin{pmatrix}C&R_B\\ 0& B\end{pmatrix}$,
where $A$ and $B$ are square matrices.
Then
    $$M_1\boxplus M_2=
    \begin{pmatrix}
    C & R_A & R_B\\
    0 & A & 0\\
    0 & 0 & B\\
    \end{pmatrix}
    $$
It is obtained that:
\begin{enumerate}
\item $\mathcal{G}(M_1 \boxplus M_2)$ is a common lift of $\mathcal{G}(M_1)$ and $\mathcal{G}(M_2)$,
\item $\max(dim(\mathcal{G}(M_1)),dim(\mathcal{G}(M_2))) \leq dim(\mathcal G(M_1\boxplus M_2))
	\leq dim(\mathcal G(M_1\oplus M_2))$.
\end{enumerate}
\end{theorem}

\begin{proof} The first item is obtained by construction. For the second one, consider
$\max(dim(\mathcal{G}(M_1)),dim(\mathcal{G}(M_2)))
	\leq \dim(\mathcal{G}(M_1 \boxplus M_2))
	=dim(\mathcal{G}(M_1)) + dim(\mathcal{G}(M_2))-dim(\mathcal G(C))
	\leq dim(\mathcal{G}(M_1)) + dim(\mathcal{G}(M_2))
	=dim(\mathcal G(M_1\oplus M_2))$
\end{proof}

Note that when the matrices $M_1$ and $M_2$ have no common columns, both  $\mathcal G(M_1\boxplus M_2)$ and $\mathcal G(M_1\oplus M_2)$ coincide. Moreover, by construction, the operation $\mathcal G(M_1\boxplus M_2)$ provides a lift which minimizes its dimension. As shown in the next example, to handle graphs using $\boxplus$ that belong to the same branch of the tree in Figure \ref{fig:arbol} have some advantages in this sense.\par

\begin{example} The first one is the hybrid graph obtained as a common lift of the $PC(2a)$ and $BCC(a)$.
The calculation described in the Theorem \ref{theo:minimumlift} leads to the matrix:

$$\begin{pmatrix}2a&0&0\\0&2a&0\\0&0&2a\end{pmatrix}
	\boxplus\begin{pmatrix}2a&0&a\\0&2a&a\\0&0&a\end{pmatrix}
	=\begin{pmatrix}2a&0&0&a\\0&2a&0&a\\0&0&2a&0\\0&0&0&a\end{pmatrix}$$

which corresponds to a 4D lattice graph. On the other hand, if we make the common lift of $PC(2a)$ and $FCC(a)$:

$$\begin{pmatrix}2a&0&0\\0&2a&0\\0&0&2a\end{pmatrix} \boxplus \begin{pmatrix}2a&a&a\\0&a&0\\0&0&a\end{pmatrix}
	=\begin{pmatrix}2a&0&0&a&a\\0&2a&0&0&0\\0&0&2a&0&0\\0&0&0&a&0\\0&0&0&0&a\end{pmatrix}$$

which generates a 5D lattice graph. In this case, the common lift has one extra dimension since the graphs
considered belong to different branches of the tree. The same happens with the mix of
$FCC(a)$ and $BCC(a)$, as shown next:

 $$\begin{pmatrix}2a&a&a\\0&a&0\\0&0&a\end{pmatrix}
	\boxplus\begin{pmatrix}2a&0&a\\0&2a&a\\0&0&a\end{pmatrix}
	=\begin{pmatrix}2a&a&a&0&a\\0&a&0&0&0\\0&0&a&0&0\\0&0&0&2a&a\\0&0&0&0&a\end{pmatrix}$$

\end{example}

Finally, to end the section we present in Table \ref{table:todasD} a selection of lattice graphs composed following the guidelines presented in this section. The table also includes their main topological characteristics. Depending on the focus some of them outperform the others.

\begin{table*}\label{table:hibridos}
    \begin{center}
 	\begin{tabular}{|llllll|}
     \hline
	Topology				        & Dimension  &Order		&Projection	  &Diameter			&Average Distance\\
	\hline
    $T(2a,2a)\boxplus RTT(a)$       &3	         &$4a^3$	&vary           &$2a$			&$\approx 1.14877a$\\
	$4D$-$FCC(a)$		                &4           &$2a^4$	&$FCC(a)$	    &$2a$			&$\approx 1.10396a$\\
    $4D$-$BCC(a)$		                &4           &$8a^4$	&$T(2a,2a,2a)$	&$2a$			&$\approx 1.5379a$\\
    $Lip(a)$		                &4           &$16a^4$	&$FCC(2a)$	    &$3a$			&$\approx 1.815a$\\
	$PC(2a)\boxplus BCC(a)$         &4		     &$8a^4$	&vary           &$2.5a$		    &$\approx 1.59715a$\\
	$PC(2a)\boxplus FCC(a)$	        &5	         &$8a^5$	&vary           &$3.5a$		    &$\approx 1.87856a$\\
	$BCC(a)\boxplus FCC(a)$         &5		     &$4a^5$	&vary           &$2.5a$	        &$\approx 1.52522a$\\
\hline
\end{tabular}
	\end{center} \caption{Distance properties of several lattice graphs}\label{table:todasD}
\end{table*}

\section{Routing in Lattice Graphs}\label{sec:routing}

Most interconnection networks use routing tables but their size can compromise system scalability. In this section routing algorithms for lattice graphs are presented. In this way, algorithmic routing can be used to avoid the need of tables. If tables are to be used, the algorithms presented can be employed to fill the routing tables.\par

Our routing algorithm is based on the hierarchy induced by the projecting operation. Routing in a lattice graph can be done by routing in its projection and in the ring defined by its side. In a first subsection we state the node labelling adopted and present a hierarchical routing. In a second subsection, we solve the routing problem in cubic crystal graphs; this is a basic contribution since we are considering cubic crystal graphs as the basic building blocks of the networks proposed in this paper. Finally, complexity and implementation aspects are considered in the last subsection.

\subsection{Hierarchical Routing}

For solving the routing problem over lattice graphs we need first to state which labelling set will be applied. A labelling set is the set which contains the labels for the vertices of the graph. There are many choices for the labelling set. In the 2D case, several approaches to the routing problem have been made in \cite{Flahive,Robic,CamareroTC}. In those articles, several labellings such as the one given by the fundamental parallelogram of the lattice, the set of integers modulo $N$ or the set of minimum norm residues have been considered. Anyway, for labelling a lattice graph of dimension $n$, a subset of $\Z^n$ will be needed. In particular, we define it as follows.

\begin{definition} Given a lattice graph $\mathcal{G}(M)$ of dimension $n$ a \textsl{labelling set} of the graph
is $\mathcal L \subset \Z^n$ such that $|\mathcal L|=|\det M|$ and for every pair $\mathbf l_1, \mathbf l_2 \in \mathcal L$ we have $\mathbf l_1 \not \equiv \mathbf l_2 \pmod{M}$.
\end{definition}

If  $\mathbf v_s, \mathbf v_d \in \mathcal L$, where $\mathbf v_s$ labels the source node and $\mathbf v_d$ labels the destination node,
we will call any vector $\mathbf r \in \mathbb{Z}^{n}$ a \textbf{routing record} when

$$\mathbf v_d - \mathbf v_s \equiv \mathbf r \pmod{M}.$$

With $\mathbf v_d-\mathbf v_s \in \Z^n$ such that:

	$$
	\mathbf v_d-\mathbf v_s\in\L-\L=\{\mathbf x-\mathbf y \ | \ \mathbf x,\mathbf y\in\L\}
	$$

From a design perspective, it is convenient to label the graph nodes according to their positive coordinates. Hence, we will consider the labelling given by the Hermite normal form of the generating matrix. Therefore, let us assume that $H$ is the Hermite normal form of $M$ and
$$
	\L=\{\mathbf x \in \mathbb{Z}^{n} \ | \ 0\leq x_i<H_{i,i}\}.
	$$
The differences set that will be the input for any of the considered routing algorithms will be:
$$
	\L-\L=\{\mathbf x \ | \ -H_{i,i}<x_i<H_{i,i}\}.
	$$

Each component of a routing record indicates the number of hops in the corresponding dimension and its sign, the direction of the hops. The length of a path associated with a routing record is given by its \textsl{Minkowski norm}:
$$
	|\mathbf r|=\sum_i|r_i|
$$
As minimal routing requires shortest paths, minimum norm routing records should be obtained. Hence, the \textbf{routing problem} over $\mathcal{G}(M)$ can be stated as follows:

\begin{center}
\begin{tabular}{rl}
\textbf{input}:		&$\mathbf v$:=$\mathbf v_{d}-\mathbf v_{s}\in\L-\L$\\
\textbf{output}:	&$\underset{\mathbf r\equiv \mathbf v \pmod{M}}{\mathrm{argmin}}(|\mathbf r|)$\\
\end{tabular}
\end{center}
where $\mathrm{argmin}$ states for the element in the set $\{\mathbf r \in \mathbb{Z}^{n} \ | \ \mathbf r\equiv \mathbf v \pmod{M}\}$ minimizing $|\mathbf r|$.\par

Our routing approach takes advantage of the hierarchical nature of lattice graphs. The idea is that routing in a lifted graph can be done by routing in its projection and in the cycle that joins the disjoint projections. Remember that the lattice graph $\mathcal{G}(M)$ with $M \cong \begin{pmatrix}B& \mathbf{c} \\0& a \end{pmatrix}$, has $a$ disjoint copies of its projection $\mathcal{G}(B)$ embedded, which are connected by $\dfrac{|\det M|}{ord(\mathbf e_n)}$ parallel cycles. The cycles have length $ord(\mathbf{e}_{n})$. The number of vertices belonging to a cycle which lie in the same copy of $\mathcal{G}(B)$ is $\dfrac{ord(\mathbf e_n)}{a}$. Hence, we can separately consider the elements of the routing record in the following way:

\begin{proposition} Let $M \cong \begin{pmatrix}B& \mathbf{c} \\0& a\end{pmatrix}$. Then, if ${\mathcal L}_{M}$ denotes the labelling set $\mathcal{G}(M)$ and ${\mathcal L}_{B}$ denotes the labelling set of its projection $\mathcal{G}(B)$ we deduce that:
$${\mathcal L}_{M} = \left \{ \begin{pmatrix} \mathbf{x} \\ y \end{pmatrix} \ | \ \mathbf{x} \in {\mathcal L}_{B}, 0 \leq y <a \right \}$$
\end{proposition}

\begin{example} The labelling of the $BCC(a)$ is:
	$$\left\{\begin{pmatrix}x\\y\\z\end{pmatrix} \mid 0\leq x<2a,\ 0\leq y<2a,\ 0\leq z<a\right\}$$
Note that it can be obtained from the labelling of the torus $T(2a,2a)$ by adding the last component $z$ fulfilling $0\leq z<a$.
\end{example}

\begin{algorithm}[t]
\SetLine
\KwIn{$\mathbf{v}_{s}$ source, $\mathbf{v}_{d}$ destination}
\KwOut{$\mathbf{r}$ minimum routing record from $\mathbf v_s$ to $\mathbf v_d$}
\SetKwFunction{Function}{Function}

Let $y$ be the last component of $\mathbf{v}_{d}$\;
$\mathbf{v}_{s} + \mathcal{C}$ is the cycle translated to $\mathbf{v}_{s}$\;
For all the vertices $\mathbf{c}_i$ of the cycle in the copy $[\mathcal{G}(B)]_{y}$ do\;
$r_{i}^{\mathcal{C}}$: Route in the cycle from $\mathbf{v}_s$ to vertex $\mathbf{c}_{i}$\;
$\mathbf{r}_{i}^{\mathcal{G}(B)}$: Route in $[\mathcal{G}(B)]_{y}$ from $\mathbf{c}_{i}$ to $\mathbf{v}_{d}$\;
Return the routing record which minimizes the weight of $\begin{pmatrix}\mathbf{r}_{i}^{\mathcal{G}(B)} \\ r_{i}^{\mathcal{C}} \end{pmatrix}$ \;
\caption{Hierarchical Routing in Lattice Graphs}\label{algo:gen-routing}
\end{algorithm}

Now, we can state the following main result:

\begin{theorem} If $[\mathcal{G}(B)]_{y}$ is the projection $\mathcal{G}(B)$ of $\mathcal{G}(M)$ which contains $y\mathbf{e}_{n}$, $\mathcal{C}$ denotes the cycle generated by $\mathbf{e}_{n}$ and, given a vertex $\mathbf{v} \in \mathbb{Z}^{n}$, $\mathbf{v} + \mathcal{C}$ denotes the translation of the cycle to this vertex. Algorithm \ref{algo:gen-routing} gives minimum routing records in any lattice graph.
\end{theorem}

\begin{proof} Since the algorithm composes routing records from two subgraphs, then the result is indeed a routing record.
We need to see that a minimum one is found.\par

Let $\mathbf r^{\min}$ be one of the routing records with minimum norm. Since $\mathbf v_s+\mathbf r^{\min}_n$ is in the cycle mentioned in the algorithm, then there is an index $i$ such that $\mathbf r^{\min}_n$ is the minimum route in the cycle from $\mathbf v_s$ to $\mathbf c_i$. As $\mathbf r^{\min}$ is minimal, we find that the minimal routing from $\mathbf{c}_{i}$ to $\mathbf{v}_{d}$ does not use the $n$ dimension. Thus, routing in $[\mathcal{G}(B)]_{y}$ gives the minimum. By composing both, the algorithm finds the minimum routing $\mathbf r^{\min}$ and returns it or another one with same norm.
\end{proof}

\begin{remark} In the last step of Algorithm \ref{algo:gen-routing} there can sometimes be
several routing records with the same weight. In this case it is advisable to choose
one of them at random, thus balancing the use of the paths.
\end{remark}

\subsection{Routing in Cubic Crystal Graphs}

We now consider specific algorithms for computing minimum routing records in cubic crystal graphs. Routing in $PC$ is widely known, since as we have previously seen, it is a 3D-torus graph. Hence, we provide routing algorithms for $FCC$ and $BCC$.\par

These new algorithms are also based on the previous idea of a hierarchical routing by using the projections of the graphs.
In general, we will denote a call to a routing algorithm in $\mathcal G(M)$ from
node $\mathbf v_s$ to node $\mathbf v_d$ as $route_M(\mathbf v_s,\mathbf v_d)$. Inside the algorithm for $\mathcal G(M)$ we will employ a nested call to $route_B$, where
$B$ is the projection $M$. This nested call can be to one of our algorithms or to
another one if it is known (as in the case of tori).\par

As we have seen previously, $FCC(a)$ defined as the lattice graph generated by
	$$\begin{pmatrix}2a&a&a\\0&a&0\\0&0&a\end{pmatrix}$$
is isomorphic to the PDTT presented in \cite{Camara}, where a generic graph routing was used.
As can be observed, its projection is the graph with matrix
$\begin{pmatrix}2a&a\\0&a\end{pmatrix}$, denoted as $RTT(a)$ in \cite{Camara}.
It is easy to verify that the order of $\mathbf e_n$ is $2a$, which implies
that the cardinal of the intersection between $\mathbf v_s+\mathcal C$ and $[\mathcal G(B)]_y$ is 2,
that is, we need to do two calls to $route_B$.
Using this mechanism we get Algorithm \ref{algo:FCC-routing} for $FCC(a)$; (for a compact notation in the Algorithm, we define the product by a Boolean as $a\cdot true=a$ and $a\cdot false=0$). An algorithm for routing in the projected 2D graph can be seen in Algorithm \ref{RTT-routing} and it has been introduced in \cite{HPC}.

\begin{algorithm}[t]
\SetLine
\KwIn{$(x,y,z)^t:=\mathbf v_{d}-\mathbf v_{s}\in\L-\L$}
\KwOut{$\mathbf{r}$ minimum routing record from $\mathbf v_s$ to $\mathbf v_d$}
\SetKwFunction{Function}{Function}

$y':=y+a(y<0)$\;
$z':=z+a(z<0)$\;
$\hat x:=x+a\bigl((y<0)\xor(z<0)\bigr)$\;
$x':=\hat x+2a(\hat x<0)-2a(\hat x\geq 2a)$\;
We have $(x',y',z')^t\in\L$\;
$\mathbf r_1^{\mathcal G(B)}:=route_{\left(\begin{smallmatrix}2a&a\\0&a\end{smallmatrix}\right)}
	(\begin{pmatrix}0\\0\end{pmatrix},\begin{pmatrix}x'\\y'\end{pmatrix})$\;
$\mathbf r_2^{\mathcal G(B)}:=route_{\left(\begin{smallmatrix}2a&a\\0&a\end{smallmatrix}\right)}
	(\begin{pmatrix}a\\0\end{pmatrix},\begin{pmatrix}x'\\y'\end{pmatrix})$\;
$\mathbf r:=\mathrm{argmin}(|\mathbf k|\mid \mathbf k\in\{
	\begin{pmatrix}\mathbf r_1^{\mathcal G(B)}\\z'\end{pmatrix},
	\begin{pmatrix}\mathbf r_2^{\mathcal G(B)}\\z'-a\end{pmatrix}
	\})$\;
\caption{Routing in $FCC(a)$}\label{algo:FCC-routing}
\end{algorithm}

\begin{remark}
When $a$ is a power of 2 the starting arithmetic operations are easier to calculate as $rem(y,a)$, $rem(z,a)$ and $rem(\hat x,2a)$.
\end{remark}

\begin{algorithm}[H]
\SetLine
\KwIn{$x,y:=v_{d}-v_{s}$}
\KwOut{$r$ routing record}
\SetKwFunction{Function}{Function}
$p:=\mathrm{rem}(x+y+a,2a)$\;
$q:=\mathrm{rem}(y-x+a,2a)$\;
$x':=(p-q)/2$\;
$y':=(p+q-2a)/2$\;
$r:=(x',y')$\;
\caption{Routing in $RTT(a)$}\label{RTT-routing}
\end{algorithm}

\begin{example} As an example let us consider $FCC(4)$. The labeling of the graph is:
	$$
	\L=\{\begin{pmatrix}x&y&z\end{pmatrix}^t:0\leq x<8,\ 0\leq y,z<4\}.
	$$
If we want to route from
$\mathbf v_s=\begin{pmatrix}1&3&3\end{pmatrix}^t$
to
$\mathbf v_d=\begin{pmatrix}6&0&1\end{pmatrix}^t$,
first we compute
$\mathbf v=\mathbf v_d-\mathbf v_s=\begin{pmatrix}5&-3&-2\end{pmatrix}^t$,
which is in  the set of differences:
	$$
	\mathbf v\in\L-\L=\{\begin{pmatrix}x&y&z\end{pmatrix}^t:-8< x<8,\ -4<y,z<4\}.
	$$
According to Algorithm \ref{algo:FCC-routing}, since we have $y=-3<0$ and $z=-2<0$ these values have to be modified as $y'=-3+4=1$ and $z'=-2+4=2$.
Moreover, since $(-3<0)\xor(-2<0)\equiv false$ we find that $\hat x=x=5$. Finally, as $0\leq 5<8$ this implies $x'=5$ and
$\mathbf v\equiv\begin{pmatrix}5&1&2\end{pmatrix}^t\in\L$.\par

Now, in $RTT(a)$ we find that a minimum route from $\begin{pmatrix}0&0\end{pmatrix}^t$ to $\begin{pmatrix}5&1\end{pmatrix}^t$ is $\begin{pmatrix}1&-3\end{pmatrix}^t$ and a minimum route from $\begin{pmatrix}4&0\end{pmatrix}^t$ to $\begin{pmatrix}5&1\end{pmatrix}^t$ is $\begin{pmatrix}1&1\end{pmatrix}^t$. Consequently,
$\mathbf r_1=\begin{pmatrix}1&-3&2\end{pmatrix}^t$
and
$\mathbf r_2=\begin{pmatrix}1&1&-2\end{pmatrix}^t$.
Finally, after comparing the two norms  $|\mathbf r_1|=6$ and $|\mathbf r_2|=4$, we find that the minimum routing record to reach $\mathbf v_d$ from $\mathbf v_s$ is given by $\mathbf r=\mathbf r_2$.
\end{example}

Similarly, for the network $BCC(a)$, we obtain Algorithm \ref{algo:BCC-routing}. Again, the order $ord(\mathbf e_n)=2a$ which implies 2 calls to the routing of a 2D torus $T(2a, 2a)$.

\begin{algorithm}[t]
\SetLine
\KwIn{$(x,y,z)^t:=\mathbf v_{d}-\mathbf v_{s}\in\L-\L$}
\KwOut{$\mathbf{r}$ minimum routing record from $\mathbf v_s$ to $\mathbf v_d$}
\SetKwFunction{Function}{Function}

$z':=z+a(z<0)$\;
$\hat x:=x+a(z<0)$\;
$\hat y:=x+a(z<0)$\;
$x':=\hat x+2a(\hat x<0)-2a(\hat x\geq 2a)$\;
$y':=\hat x+2a(\hat y<0)-2a(\hat y\geq 2a)$\;
We have $(x',y',z')^t\in\L$\;
$\mathbf r_1^{\mathcal G(B)}:=route_{\left(\begin{smallmatrix}2a&0\\0&2a\end{smallmatrix}\right)}
	(\begin{pmatrix}0\\0\end{pmatrix},\begin{pmatrix}x'\\y'\end{pmatrix})$\;
$\mathbf r_2^{\mathcal G(B)}:=route_{\left(\begin{smallmatrix}2a&0\\0&2a\end{smallmatrix}\right)}
	(\begin{pmatrix}a\\a\end{pmatrix},\begin{pmatrix}x'\\y'\end{pmatrix})$\;
$\mathbf r:=\mathrm{argmin}(|\mathbf k|\mid \mathbf k\in\{
	\begin{pmatrix}\mathbf r_1^{\mathcal G(B)}\\z'\end{pmatrix},
	\begin{pmatrix}\mathbf r_2^{\mathcal G(B)}\\z'-a\end{pmatrix}
	\})$\;

\caption{Routing in $BCC(a)$}\label{algo:BCC-routing}
\end{algorithm}

\subsection{Routing Discussion}

Routing in circulant graphs was first related to the Shortest Vector Problem (SVP) in \cite{Cai}. Later, this fact was used to optimize a routing algorithm for circulants of degree four in \cite{Gomez-Perez}. Following the same ideas, similar complexity for the SVP can be inferred for routing in lattice graphs. However, algorithms for particular graphs can be improved. In this subsection we consider how to appropriately choose the projection of the lattice graph in order to obtain the best routing algorithm among all the possibilities.\par

First, note that following the ideas in the previous section, we can infer the impact of routing complexity for the different lifts of crystal lattice graphs. As we have seen, $\dfrac{ord(\mathbf e_n)}{a}$ determines the number of intersections of the cycle with the destination projection,
which dictates the number of nested routing calls.

\begin{remark}
In fact, we find that:
\begin{itemize}
\item The routing in $nD$-$PC$ can be done immediately with $n$ comparisons in parallel.
\item The hierarchical routing in $nD$-$BCC$ requires 2 calls to the routing algorithm for $(n-1)D$-$PC$.
\item The hierarchical routing in $nD$-$FCC$ requires 2 calls to the $(n-1)D$-$FCC$, which accumulates
	into $2^{n-2}$ calls to $2D$-$FCC$ or RTT. These last routing calls will be performed by Algorithm \ref{RTT-routing}.
\end{itemize}
\end{remark}

Finally, let us consider the case of hybrid graphs. As we have seen in Section \ref{sec:lifting}, hybrid graphs are obtained as common lifts of different lattice graphs. Therefore, given a hybrid graph $\mathcal{G}(M)$ there would be several possible lattice graphs which could be considered as its projection. Since the heaviest computation part in Algorithm \ref{algo:gen-routing} corresponds to the routing calls in the projection, that projection should be carefully chosen. For example, let $\mathcal{G}(M)$ be such that:

$$\begin{pmatrix}2a&0&0&a\\0&2a&0&a\\0&0&2a&0\\0&0&0&a\end{pmatrix}$$

As we have previously seen, this graph is obtained as the common lift of $PC(2a)$ and $BCC(a)$. Clearly, taking $BCC(a)$ as the projection, would complicate the routing function. Hence, we should choose $PC(2a)$ as its projection, in which dependencies among dimensions do not exist and routing will be less laborious.

\section{Practical Issues}\label{sec:practical}

This work has been conceived to study the fundamentals of twisting wrap-around links in multidimensional torus networks. Nevertheless, this research has been motivated by the widespread presence of moderate degree tori in the supercomputing market. Although Fujitsu has recently entered in this terrain with its K system, traditionally Cray and IBM are the two major companies standing out for years in the development of interconnection networks based on torus networks. Hence, this section will be devoted to discuss certain practical aspects. The first one is related to physical network deployment and the second consists of a preliminary performance evaluation.

\subsection{Physical Organization}

It is not difficult to conceive a package hierarchy and a 3D physical organization to deploy systems based on lattice graphs. For illustrating this organization, let us first consider the approach followed by manufacturers. Cray uses a straightforward structure. For example,
an actual configuration, \cite{jaguar}, was a $T(25,32,16)$ packaged on a 200 rack system arranged as an $8 \times25$ rectangle. We can see the system as:
\begin{itemize}
\item System of $25\times 8\times 1$ racks.
\item Racks of $1\times 4 \times 16$ nodes.
\end{itemize}
That is, the third dimension is completely inside the racks and the first dimension is formed entirely joining racks. However the second dimension is partially inside the rack and requires connecting rack columns by rings. Taking into account forthcoming improvements in integration and packaging technologies, it could be expected that a 4D torus would have two dimensions internal to the racks and the other 2 external to the racks.
This idea generalizes to lattice graphs. If $\mathcal G(M)$ is a 4D lattice graph, its 2D projections
would be built inside racks, which would be a torus or a twisted torus.
Then it becomes a matter of completing the lattice by adjusting the offsets
of the cables connecting the racks. Moreover, folding techniques for 3D networks presented in \cite{Camara} can also be of application in our case and easily generalized to higher dimensions.\par

IBM presents a more elaborated organization in the Blue Gene family, \cite{Coteus}. Although the complete network is a torus,
each midplane (half of a rack) has additional hardware which enables the midplane
to disconnect from the remainder of the network and to be itself a small torus.
By arranging several midplanes, this additional hardware enables a multitude of different tori shapes to be connected.
With slight modifications to such hardware it is possible to allow each group to
be a symmetric crystal (or another lattice if desired) instead of a mixed-radix torus.
This hardware changes its configuration only between different application runs. Then, the potentially added functionality would not have
any negative impact on the system.

\subsection{Evaluation compared to currently used topologies}

Most evaluations of big networks have relied on measuring their behavior when managing synthetic traffic loads. Typical experiments are based on simulation. Notwithstanding, the work presented in \cite{Chen} evaluates different routing algorithms reporting maximum achievable loads on a real IBM BlueGene system. They make runs on machines whose topologies are the torus $T(8,8,8,4,2)$ and $T(16,8,8,8,2)$. We shall ignore the last dimension of size 2 and treat them as four dimensional networks; the last small dimension comes from the inside of computing nodes, fixed by computer technology. We have simulated the same tori plus symmetric lattice graphs of the same sizes. We evaluate $4D$-$BCC(4)$ compared to $T(8,8,8,4)$ and $4D$-$FCC(8)$ compared to $T(16,8,8,8)$.

We have used the same synthetic traffic patterns as in \cite{Chen}:

\begin{itemize}
\item \textbf{uniform:} Each node generates packets to any other node, at random with a uniform probability distribution.
\item \textbf{antipodal:} Each node generates traffic to the most distant one.
\item \textbf{centralsymmetric:} Once a center of symmetry is fixed, each node has as its destination the symmetric one.
\item \textbf{randompairings:} The network is divided into pairs in a random uniform way, which then communicate for all the simulation.
\end{itemize}

\begin{table}[b]
	\begin{center}
	\small
	\begin{tabular}{|l|c|}
		\hline
		Injectors & 6   \\
		\hline
		Packet size & 16 phits \\
        \hline
        Queues & 4 packets \\
		\hline
		Deadlock avoidance & Bubble \\
		\hline
		Virtual Channels & 3 \\
        \hline
        flow control & Virtual Cut-through \\
        \hline
        Routing Mechanisms & DOR \\
        \hline
        Arbitration mechanism & random \\
		\hline
	\end{tabular}
	\end{center}
	\caption{Simulation parameters}
	\label{table:datos}
\end{table}

Simulations have been conducted using INSEE (Interconnection Network Simulation and Evaluation Environment) \cite{Navaridas}. Their basic units are the \emph{cycle} for measuring time and the \emph{phit} for measuring information. Each network link (edge of the graph) can send one or zero phits in each cycle. The network \emph{load} is the amount of information received per time. We measure the network load in $\text{phits}/(\text{cycle}\cdot \text{node})$. Nodes (vertices of the graph) generate \emph{packets} composed of an integral number of phits (typically constant) to be sent to other network nodes. For any provided traffic up to load $l$, a packet is injected each cycle in each node with probability $\dfrac{l}{s}$, where $s$ is the size of a packet measured in phits. The accepted traffic or throughput is the total of phits received, divided by the total simulation time and by the number of nodes $N$.
Simulation parameters are shown in Table \ref{table:datos}. We have simulated $10,000$ cycles for statistics, preceded by a network warmup.
At least 5 simulations are averaged for each point. The BlueGene family of supercomputers implements a congestion control mechanism that prioritizes in-transit traffic over new injections, which is also implemented in our router model.

\PEAKBARSLARGE
\PEAKBARSSMALL

Figures \ref{fig:simulation_large_peaks} and \ref{fig:simulation_small_peaks} show results of accepted load in the four networks. Under uniform traffic, results exhibit gains of 26\% in the small case ($4D$-$BCC$) and 50\% in the large one ($4D$-$FCC$).
In random pairings, the throughput is consistently higher, with gains of 16\% and 2\% respectively.
The other two traffic patterns show congestion at high loads for all the networks considered. Nevertheless, the peak load for the antipodal traffic improves by 62\% and 75\% respectively. Under central symmetric traffic, gains are of 45\% in the small case and of 23\% in the large one. Figures \ref{fig:simulation_latency_large} and \ref{fig:simulation_latency_small} show average packet latencies. The different curves demonstrate the superior behavior of lattice topologies. Gain values are coherent with the topological analysis presented in Subsection \ref{sub:crystal_comp}.

\SIMULATIONLATENCYLARGE
\SIMULATIONLATENCYSMALL

\section{Conclusions}\label{sec:conclu}

This research has been focused on the study and proposal of new multidimensional twisted torus interconnection networks. Due to their complex spatial characteristics, their analysis is far from straightforward. Nevertheless, we have taken advantage of an algebraic tool based on integral square matrices presented in \cite{Fiol}. Such matrices define the graph and its topological characteristics. Adequate algebraic manipulations of the matrices enable a better understanding of different network properties. For example, when using the Hermite normal form, matrices reveal the subgraphs naturally embedded in the network.\par

Using this tool, several networks have been proposed and analyzed in this paper. We firstly focus on 3D symmetric networks as alternatives to mixed-radix tori which are not edge-symmetric. Taking the matrices that define cubic crystallographic lattices, we were able to evaluate and compare their associated interconnection networks. If symmetry is desired, the best path when upgrading 3D systems clearly seems to be $PC(a)\to FCC(a)\to BCC(a)\to PC(2a)$, that is, duplicating the machine size on each step and maintaining most of the original connections. In addition, we have introduced a couple of graph lifting methods that allow for higher dimensional networks that embed cubic crystal subnetworks among other graphs. Complementarily, the use of graph projections facilitates the conception of routing algorithms for these networks. Based on this graph operation, minimal routing schemes have been proposed for all the topologies. Although we have focused on typical network configurations derived from powers of two, our results remain valid for any other network size.\par

The paper preliminarily addresses some practical issues. Physical packaging and system organization in racks have been considered, concluding that, for deploying networks based on lattice graphs, very few changes over typical tori would be necessary. In addition to the algebraic analysis carried out through the paper, an empirical evaluation of different interesting topologies has been carried out. Comparisons with current machines have certified that multidimensional twisted tori clearly outperform their orthogonal counterparts. Noticeable gains were exhibited by twisted lattice topologies for both configurations under consideration. These preliminary experiments motivate a thorough network evaluation that will be reported in a forthcoming work.

\appendix
\section{Symmetric Lattice Graphs of dimension 3}\label{apendice}

This Appendix provides a complete characterization of those lattice graphs which are edge-symmetric by linear automorphisms.
In Subsection \ref{sub:linear} some definitions and preliminary lemmas are obtained. In Subsection \ref{sub:characterization},
the complete characterization is done. Finally, some additional comments on the non-linear case are done in Subsection \ref{sub:non-linear}. 

\subsection{About linear automorphisms}
\label{sub:linear}

\begin{definition}A \textsl{signed permutation} of length $n\in\mathbb N$ is a composition of a
sign changing function ($k\to\pm k$, $1\leq k\leq n$) and a permutation $\pi\in\Sigma_n$.

Then, we call \textsl{signed permutation matrix} to a matrix such that when it multiplies a vector
it applies the signed permutation to the vector. Signed permutation matrices are the matrices 
such that in each row and column all entries are zero except exactly one entry with value $\pm 1$.
\end{definition}

In \cite{IWONT10} the two following results were proved.

\begin{lemma}For any linear automorphism $\phi$ of $\mathcal G(M)$ with $\phi(0)=0$ there exists a
signed permutation matrix $P$ such that $\phi(\mathbf x)=P\mathbf x$.
\end{lemma}

\begin{lemma} The function $\phi$ defined by $\phi(\mathbf x)=P\mathbf x$ is an automorphism
of $\mathcal G(M)$ if and only if there exists $Q\in\mathbb\Z^{n\times n}$
such that $PM=MQ$.
\end{lemma}

The linear automorphisms of a lattice graph $\mathcal G(M)$ form a group $LAut(\mathcal G(M))$,
which usually coincides with the full automorphism group $Aut(\mathcal G(M))$,
except in a few cases that we consider in the last section of this appendix.
The group of linear automorphisms which fixes 0 will be denoted as $LAut(\mathcal G(M),0)$
(also known as \textsl{stabilizer}).

\begin{definition}We say that $\mathcal G(M)$ is linearly-symmetric if for every $i$
there exist $\phi\in LAut(\mathcal G(M),0)$ such that $\phi(e_1)=\pm e_i$.
\end{definition}

\begin{lemma} A linearly-symmetric lattice graph is symmetric.
\end{lemma}

We can denote signed permutations as $(1\ -\!2)(-\!3\ -\!4)$, where
$\sigma=(\dots\ \pm\!a\ b\ \dots)$ means that $\sigma(a)=b=-\sigma(-a)$ and
$\sigma=(\dots\ \pm\!a\ -\!b\ \dots)$ means that $\sigma(a)=-b=-\sigma(-a)$.
The number of signed permutations of length $n$ is $n!2^n$.
For $n=3$ this is $3!2^3=48$, which are given in Table \ref{table:signedpermutations}.
\begin{table}
\begin{center}\scriptsize
$
\begin{array}{|c|c|c|c|}
\hline
\rule{0pt}{1em}
\sigma			&\sigma^2			&\sigma^3			&\sigma^4\\
\hline
id=(1)(2)(3)	&id					&					&id\\
(1)(2)(-\!3)		&id					&					&id\\
(1)(-\!2)(3)		&id					&					&id\\
(1)(-\!2)(-\!3)		&id					&					&id\\
(-\!1)(2)(3)		&id					&					&id\\
(-\!1)(2)(-\!3)		&id					&					&id\\
(-\!1)(-\!2)(3)		&id					&					&id\\
(-\!1)(-\!2)(-\!3)	&id					&					&id\\

\hline
(1)(2\ 3)		&id					&					&id\\
(1)(2\ -\!3)		&(1)(-\!2)(-\!3)		&					&id\\
(1)(-\!2\ 3)		&(1)(-\!2)(-\!3)		&					&id\\
(1)(-\!2\ -\!3)		&id					&					&id\\
(-\!1)(2\ 3)		&id					&					&id\\
(-\!1)(2\ -\!3)		&(1)(-\!2)(-\!3)		&					&id\\
(-\!1)(-\!2\ 3)		&(1)(-\!2)(-\!3)		&					&id\\
(-\!1)(-\!2\ -\!3)		&id					&					&id\\

\hline
(1\ 3)(2)		&id					&					&id\\
(1\ -\!3)(2)		&(-\!1)(2)(-\!3)		&					&id\\
(1\ 3)(-\!2)		&id					&					&id\\
(1\ -\!3)(-\!2)		&(-\!1)(2)(-\!3)		&					&id\\
(-\!1\ 3)(2)		&(-\!1)(2)(-\!3)		&					&id\\
(-\!1\ -\!3)(2)		&id					&					&id\\
(-\!1\ 3)(-\!2)		&(-\!1)(2)(-\!3)		&					&id\\
(-\!1\ -\!3)(-\!2)		&id					&					&id\\

\hline
(1\ 2)(3)		&id					&					&id\\
(1\ 2)(-\!3)		&id					&					&id\\
(1\ -\!2)(3)		&(-\!1)(-\!2)(3)		&					&id\\
(1\ -\!2)(-\!3)		&(-\!1)(-\!2)(3)		&					&id\\
(-\!1\ 2)(3)		&(-\!1)(-\!2)(3)		&					&id\\
(-\!1\ 2)(-\!3)		&(-\!1)(-\!2)(3)		&					&id\\
(-\!1\ -\!2)(3)		&id					&					&id\\
(-\!1\ -\!2)(-\!3)		&id					&					&id\\

\hline
(1\ 2\ 3)			&(1\ 3\ 2)			&id					&(1\ 2\ 3)\\
(1\ 2\ -\!3)		&(-\!1\ -\!3\ 2)			&(-\!1)(-\!2)(-\!3)		&(-\!1\ -\!2\ 3)\\
(1\ -\!2\ 3)		&(1\ -\!3\ -\!2)			&(-\!1)(-\!2)(-\!3)		&(-\!1\ 2\ -\!3)\\
(1\ -\!2\ -\!3)		&(-\!1\ 3\ -\!2)			&id					&(1\ -\!2\ -\!3)\\
(-\!1\ 2\ 3)		&(-\!1\ 3\ -\!2)			&(-\!1)(-\!2)(-\!3)		&(1\ -\!2\ -\!3)\\
(-\!1\ 2\ -\!3)		&(1\ -\!3\ -\!2)			&id					&(-\!1\ 2\ -\!3)\\
(-\!1\ -\!2\ 3)		&(-\!1\ -\!3\ 2)			&id					&(-\!1\ -\!2\ 3)\\
(-\!1\ -\!2\ -\!3)		&(1\ 3\ 2)			&(-\!1)(-\!2)(-\!3)		&(1\ 2\ 3)\\

\hline
(1\ 3\ 2)			&(1\ 2\ 3)			&id					&(1\ 3\ 2)\\
(1\ -\!3\ 2)		&(1\ -\!2\ -\!3)			&(-\!1)(-\!2)(-\!3)		&(-\!1\ 3\ -\!2)\\
(1\ 3\ -\!2)		&(-\!1\ -\!2\ 3)			&(-\!1)(-\!2)(-\!3)		&(-\!1\ -\!3\ 2)\\
(1\ -\!3\ -\!2)		&(-\!1\ 2\ -\!3)			&id					&(1\ -\!3\ -\!2)\\
(-\!1\ 3\ 2)		&(-\!1\ 2\ -\!3)			&(-\!1)(-\!2)(-\!3)		&(1\ -\!3\ -\!2)\\
(-\!1\ -\!3\ 2)		&(-\!1\ -\!2\ 3)			&id					&(-\!1\ -\!3\ 2)\\
(-\!1\ 3\ -\!2)		&(1\ -\!2\ -\!3)			&id					&(-\!1\ 3\ -\!2)\\
(-\!1\ -\!3\ -\!2)		&(1\ 2\ 3)			&(-\!1)(-\!2)(-\!3)		&(1\ 3\ 2)\\

\hline
\end{array}
$
\end{center}
	\caption{Signed permutations of 3 elements and their powers}
	\label{table:signedpermutations}
\end{table}

\subsection{Determination of all linearly symmetric lattice graphs for $n=3$}\label{sub:characterization}

\begin{definition}A pair of matrices $A,B\in\mathbb{\Z}^{n\times n}$ are \textsl{similar}
when a unit matrix $U$ exists such that $AU=UB$. This is denoted
by $A\sim B$.
\end{definition}

\begin{lemma}\label{lem:rightsimilar}
%
%
Let $PM=MQ$ and $PM'=M'Q'$ then
	$M\cong M'\Leftrightarrow Q\sim Q'.$
\end{lemma}
\begin{proof}
We see that if we know $PM=MQ$ and $M=M'U$ then $PM'U=M'UQ$ and $PM'=M'(UQU^{-1})=M'Q'$
with $Q'\sim Q$. Reciprocally, we know that if $PM=MQ$ and $Q'=UQU^{-1}$
then $M'=MU$ produces $PM'=M'Q'$ and $M'\cong M$.
\end{proof}
Since right equivalences leave the group invariant (hence the graph is the same),
we know that for a given $P$ we only need to see how many $Q$ there are
modulo similarity. Then, knowing $P$ and $Q$ we can
solve for $M$.


In \cite{Newman} the following useful theorem is stated:
\begin{theorem}\label{thm:newmanreducible}
Given a matrix $A$ we can find a similar matrix, made of blocks,
which is block upper triangular and moreover, that the blocks of the diagonal
all have characteristic polynomial irreducible over $\Q$ (Theorem III.12, page 50).
\end{theorem}


One simple case is when $LAut(\mathcal G(M),0)$ is a cyclic group $\langle \phi\rangle$.
In this case the associated matrix will have characteristic polynomial $x^n\pm 1$.
Starting at $n=4$ we can find groups, such as the Klein four-group in which the group
is generated by more than 1 element.

\begin{lemma}\label{lem:symorder3} Given $M\in\Z^{3\times 3}$, $\mathcal G(M)$ is linearly symmetric
if and only if there exists a signed permutation of order 3 in $LAut(\mathcal G(M),0)$.
\end{lemma}
\begin{proof}

If such a signed permutation exists, it is clear that $\mathcal G(M)$ is linearly symmetric.

For the reciprocal, we begin noting that signed permutations of length 3 can have orders 1, 2, 3, 4 and 6.
The identity is the only signed permutation of order 1 and does not contribute to symmetry.
Moreover, the signed permutations which only change signs (such as $(-1)(2)(-3)$)
do no contribute to symmetry.
Any remaining signed permutation of orders 2 and 4 do not provide symmetry by themselves,
and the composition of two of them generates either a sign change or a permutation of
order 3 or 6.

Hence linear symmetry implies the existence of an automorphism $\phi\in LAut(\mathcal G(M,0))$
with order 3 or 6. If it has order 3, we already have the desired permutation.
Otherwise we have $\phi^3=-id$ and so $\psi=\phi^2$ has order 3.
\end{proof}

Hence, if $\mathcal G(M)$ is linearly symmetric then $LAut(\mathcal G(M),0)$
contains at least one of the next four groups as a subgroup and
there is a matrix $P$ such that $PM=MQ$ for some $Q$.
\begin{itemize}
\item $\langle(1\ 2\ 3)\rangle=\langle(1\ 3\ 2)\rangle$ with $P_1$.
\item $\langle(1\ -\!2\ -\!3)\rangle=\langle(-\!1\ 3\ -\!2)\rangle$ with $P_2$.
\item $\langle(-\!1\ 2\ -\!3)\rangle=\langle(1\ -\!3\ -\!2)\rangle$ with $P_3$.
\item $\langle(-\!1\ -\!2\ 3)\rangle=\langle(-\!1\ -\!3\ 2)\rangle$ with $P_4$.
\end{itemize}
\begin{align*}
	P_1&=\begin{pmatrix}0&0&1\\1&0&0\\0&1&0\end{pmatrix}&
	P_2&=\begin{pmatrix}0&0&1\\-1&0&0\\0&-1&0\end{pmatrix}\\
	P_3&=\begin{pmatrix}0&0&-1\\1&0&0\\0&-1&0\end{pmatrix}&
	P_4&=\begin{pmatrix}0&0&-1\\-1&0&0\\0&1&0\end{pmatrix}
\end{align*}
These signed permutations have characteristic and minimum polynomial $x^3-1$.
We can find some matrices (symbolic over 3 integer parameters)
by taking $Q=P$, that is, we obtain $M_i$ such that $P_iM_i=M_iP_i$. They are:
\begin{align*}
M_1&=\begin{pmatrix}a&c&b\\b&a&c\\c&b&a\end{pmatrix},&
M_2&=\begin{pmatrix}a&-c&-b\\b&a&-c\\c&b&a\end{pmatrix},\\
M_3&=\begin{pmatrix}a&-c&-b\\b&a&c\\c&-b&a\end{pmatrix},&
M_4&=\begin{pmatrix}a&c&b\\b&a&-c\\c&-b&a\end{pmatrix}.
\end{align*}
We now need to find the similar matrices.\\

\begin{lemma}\label{lem:similarclasses}
There are exactly 2 similarity classes with characteristic polynomial $x^3-1$:
$$
Q_1=\begin{pmatrix}1&0&0\\0&-1&1\\0&-1&0\end{pmatrix}\text{ and }
Q_2=\begin{pmatrix}1&0&1\\0&-1&1\\0&-1&0\end{pmatrix}.
$$
\end{lemma}
\begin{proof}
For $x^3-1=(x-1)(x(x+1)+1)$ we have the following upper triangular block matrix
which has it as its characteristic polynomial:
$Q=\begin{pmatrix}1&0&0\\0&-1&1\\0&-1&0\end{pmatrix}$.
We know that
\begin{multline*}
\begin{pmatrix}1&n&m\\0&1&0\\0&0&1\end{pmatrix}
\begin{pmatrix}1&m+2n&m-n\\0&-1&1\\0&-1&0\end{pmatrix}
\begin{pmatrix}1&-n&-m\\0&1&0\\0&0&1\end{pmatrix}\\
=
\begin{pmatrix}1&0&0\\0&-1&1\\0&-1&0\end{pmatrix}
\end{multline*}
So
	$
\begin{pmatrix}1&m+2n&m-n\\0&-1&1\\0&-1&0\end{pmatrix}
\sim
\begin{pmatrix}1&0&0\\0&-1&1\\0&-1&0\end{pmatrix}
	$.
And as $|\det(\begin{pmatrix}-1&-2\\-1&1\end{pmatrix})|=3$,
by Theorem \ref{thm:newmanreducible},
we have at most 3 matrices modulo similarity, which are:
$$
\begin{pmatrix}1&0&0\\0&-1&1\\0&-1&0\end{pmatrix},
\begin{pmatrix}1&0&1\\0&-1&1\\0&-1&0\end{pmatrix}\text{ and }
\begin{pmatrix}1&0&2\\0&-1&1\\0&-1&0\end{pmatrix}.
$$
We check that the first two are non-similar.
If
$$
\begin{pmatrix}1&0&0\\0&-1&1\\0&-1&0\end{pmatrix}
\begin{pmatrix}a&b&c\\d&e&f\\g&h&i\end{pmatrix}
=
\begin{pmatrix}a&b&c\\d&e&f\\g&h&i\end{pmatrix}
\begin{pmatrix}1&0&1\\0&-1&1\\0&-1&0\end{pmatrix}
$$
then
$$
\begin{pmatrix}a&b&c\\-d-g&-e+h&-f+i\\-d&-e&-f\end{pmatrix}
=
\begin{pmatrix}a&-b-c&a+b\\d&-e-f&d+e\\g&-h-i&g+h\end{pmatrix}.
$$
Hence $d=g=0$ and $a=-3b$, and $3b$ divides the determinant, which cannot be a unit.
Now we see that the last two are similar.
$$
\begin{pmatrix}1&0&1\\0&-1&1\\0&-1&0\end{pmatrix}
\begin{pmatrix}1&0&1\\0&0&1\\0&-1&1\end{pmatrix}
=
\begin{pmatrix}1&0&1\\0&0&1\\0&-1&1\end{pmatrix}
\begin{pmatrix}1&0&2\\0&-1&1\\0&-1&0\end{pmatrix}
$$
So we have proved that there are exactly 2 similarity classes
with characteristic polynomial $x^3-1$:
$$
Q_1=\begin{pmatrix}1&0&0\\0&-1&1\\0&-1&0\end{pmatrix}\text{ and }
Q_2=\begin{pmatrix}1&0&1\\0&-1&1\\0&-1&0\end{pmatrix}.
$$
\end{proof}

We need to explore the $4\cdot 2=8$ possible matrices from all combinations.

\begin{lemma}\label{lem:similarities}We have $P_1\sim Q_2\sim P_2\sim P_3 \sim P_4$.
\end{lemma}
\begin{proof}
First we see that $P_1\sim Q_2$.
$$
\begin{pmatrix}0&0&1\\1&0&0\\0&1&0\end{pmatrix}
\begin{pmatrix}1&0&0\\1&-1&1\\1&0&1\end{pmatrix}
=
\begin{pmatrix}1&0&0\\1&-1&1\\1&0&1\end{pmatrix}
\begin{pmatrix}1&0&1\\0&-1&1\\0&-1&0\end{pmatrix}
$$
And now that $P_1\sim P_2\sim P_3 \sim P_4$.
$$
\begin{pmatrix}0&0&1\\1&0&0\\0&1&0\end{pmatrix}
\begin{pmatrix}-1&0&0\\0&1&0\\0&0&-1\end{pmatrix}
=
\begin{pmatrix}-1&0&0\\0&1&0\\0&0&-1\end{pmatrix}
\begin{pmatrix}0&0&1\\-1&0&0\\0&-1&0\end{pmatrix}
$$
$$
\begin{pmatrix}0&0&1\\1&0&0\\0&1&0\end{pmatrix}
\begin{pmatrix}1&0&0\\0&1&0\\0&0&-1\end{pmatrix}
=
\begin{pmatrix}1&0&0\\0&1&0\\0&0&-1\end{pmatrix}
\begin{pmatrix}0&0&-1\\1&0&0\\0&-1&0\end{pmatrix}
$$
$$
\begin{pmatrix}0&0&1\\1&0&0\\0&1&0\end{pmatrix}
\begin{pmatrix}1&0&0\\0&-1&0\\0&0&-1\end{pmatrix}
=
\begin{pmatrix}1&0&0\\0&-1&0\\0&0&-1\end{pmatrix}
\begin{pmatrix}0&0&-1\\-1&0&0\\0&1&0\end{pmatrix}
$$
\end{proof}
Thus, the first 4 matrices with $P_iM=MQ_2$ are right equivalent to the
previously calculated $M_i$.

Now we find the 4 symbolic matrices $M_i'$ which satisfy $P_iM_i'=M_i'Q_1$.

\begin{align*}
M_1'&=\begin{pmatrix}a&b&c\\a&c&-b-c\\a&-b-c&b\end{pmatrix}\\
M_2'&=\begin{pmatrix}a&b&c\\-a&-c&b+c\\a&-b-c&b\end{pmatrix}\\
M_3'&=\begin{pmatrix}a&b&c\\a&c&-b-c\\-a&b+c&-b\end{pmatrix}\\
M_4'&=\begin{pmatrix}a&b&c\\-a&-c&b+c\\-a&b+c&-b\end{pmatrix}
\end{align*}

The next two lemmas show that the 8 families of matrices modulo similarity
are actually only 2 families modulo graph isomorphism.
\begin{lemma}\label{lem:Miso}
The sets induced by the matrices $M_1$, $M_2$, $M_3$ and $M_4$
are the same modulo graph-isomorphism
when taking the parameters $a,b,c\in\mathbb Z$.
\end{lemma}
\begin{proof}
$$
\begin{pmatrix}-1&0&0\\0&1&0\\0&0&1\end{pmatrix}
M_1
\begin{pmatrix}1&0&0\\0&-1&0\\0&0&-1\end{pmatrix}
=
\begin{pmatrix}-a&c&b\\b&-a&-c\\c&-b&-a\end{pmatrix}
$$
which is $M_4$ giving $a$ the value $-a$.

$$
\begin{pmatrix}1&0&0\\0&-1&0\\0&0&1\end{pmatrix}
M_1
\begin{pmatrix}-1&0&0\\0&1&0\\0&0&-1\end{pmatrix}
=
\begin{pmatrix}-a&c&-b\\b&-a&c\\-c&b&-a\end{pmatrix}
$$
which is $M_2$ giving $a$ the value $-a$ and $c$ the value $-c$.

$$
\begin{pmatrix}1&0&0\\0&1&0\\0&0&-1\end{pmatrix}
M_1
\begin{pmatrix}1&0&0\\0&1&0\\0&0&-1\end{pmatrix}
=
\begin{pmatrix}a&c&-b\\b&a&-c\\-c&-b&a\end{pmatrix}
$$
which is $M_3$ giving $c$ the value $-c$.
\end{proof}

\begin{lemma}\label{lem:Mprimeiso}
The sets induced by the matrices $M_1'$, $M_2'$, $M_3'$ and $M_4'$
are the same modulo graph-isomorphism
when taking the parameters $a,b,c\in\mathbb Z$.
\end{lemma}
\begin{proof}
\begin{multline*}
M_1'
=\begin{pmatrix}1&0&0\\0&-1&0\\0&0&1\end{pmatrix}M_2'
=\begin{pmatrix}1&0&0\\0&1&0\\0&0&-1\end{pmatrix}M_3'
=\\\begin{pmatrix}1&0&0\\0&-1&0\\0&0&-1\end{pmatrix}M_4'
\end{multline*}
\end{proof}

\begin{theorem}Any linearly symmetric lattice graph of dimension 3 is isomorphic to another generated by one of the matrices
$$
M_1=\begin{pmatrix}a&c&b\\b&a&c\\c&b&a\end{pmatrix}\text{ or }
M_1'=\begin{pmatrix}a&b&c\\a&c&-b-c\\a&-b-c&b\end{pmatrix}
$$
for some $a,b,c\in\mathbb Z$.
\end{theorem}
\begin{proof}Let a linearly symmetric lattice graph $\mathcal G(M)$ with $M\in\mathbb Z^{3\times 3}$.
By Lemma \ref{lem:symorder3}, $P$ must exist with $PM=MQ$ with $P\in\{P_1,P_2,P_3,P_4\}$.
By Lemmas \ref{lem:rightsimilar} and \ref{lem:similarclasses} there exist $M'$ and $Q$
with $M\cong M'$, $Q\in\{Q_1,Q_2\}$ and $PM'=M'Q$.
If $Q=Q_2$, then by Lemma \ref{lem:similarities} we know $M''\in\{M_1,M_2,M_3,M_4\}$
with $PM''=M''P$, $M''\cong M$,
by Lemma \ref{lem:Miso}, it follows that $\mathcal G(M)\cong\mathcal G(M_1)$.
If $Q=Q_1$, then by Lemma \ref{lem:similarities} we know $M'\in\{M_1',M_2',M_3',M_4'\}$,
thus by Lemma \ref{lem:Mprimeiso} we obtain that $\mathcal G(M)\cong\mathcal G(M_1')$.
\end{proof}

\subsection{Non-linear automorphisms}\label{sub:non-linear}
In some cases, there are no linear automorphisms which give symmetry, although
some non-linear automorphisms do so.
The following theorem first stated in \cite{IWONT10} analyzes those cases.

\begin{definition} We say that $\mathbf a, \mathbf b, \mathbf c, \mathbf d \in \pm \mathcal B_n$ form a 4-cycle in $\mathcal G(M)$ if
$0=\mathbf a+\mathbf b+\mathbf c+\mathbf d$\footnote{each of $\{(\mathbf v,\mathbf v+\mathbf a,\mathbf v+\mathbf a+\mathbf b,\mathbf v+\mathbf a+\mathbf b+\mathbf c,\mathbf v+\mathbf a+\mathbf b+\mathbf c+\mathbf d):\mathbf v\in \mathbf G(M)\}$ is a cycle of length 4}.
Then, we say that $\mathcal G(M)$ has no nontrivial 4-cycles if  $\mathbf a,\mathbf b,\mathbf c,\mathbf d\in \pm\mathcal B_n$
such that $0=\mathbf a+\mathbf b+\mathbf c+\mathbf d$ which implies $\mathbf a=-\mathbf b$ or $\mathbf a=-\mathbf c$ or $\mathbf a=-\mathbf d$.
\end{definition}

\begin{theorem}\label{theo:linear} If the connected lattice graph $\mathcal G(M)$ has no
nontrivial 4-cycles then any graph automorphism with
$\phi(0)=0$ is a group automorphism of $\Z^n/M\Z^n$.
\end{theorem}

For $n=2$ all symmetric lattice graphs which are not linearly symmetric were determined, which are:
\begin{itemize}
\item The ones which had two linearly independent nontrivial 4-cycles.
\item The ones with exactly one nontrivial 4-cycle.
\end{itemize}
The first item directly produces the matrices (plus their divisors), since
they are combinations of $(4,0),(3,1),(2,2)$ with the appropriate changes.
For the second item, it was seen that the only ones were the family
$\begin{pmatrix}m&2\\n&2\end{pmatrix}$, which are the only ones which fail Adam-isomorphy \cite{Zerovnik}.\\

For more dimensions, first we note all the possible nontrivial 4-cycles (up to adding zeroes and sign permuting):
\begin{itemize}
\item $(4)$, first appearing at dimension $n=1$
\item $(3,1)$, $(2,2)$ first appearing at $n=2$
\item $(2,1,1)$ first appearing at $n=3$
\item $(1,1,1,1)$ first appearing at $n=4$
\end{itemize}
Symmetric graphs which are not linearly symmetric lattice graphs can be obtained
by using one of the 4-cycles as a column, completing the matrix
and checking if the matrix or one of its divisors generates a symmetric lattice graph.
Here we will not perform the complete characterization of the symmetric lattice graph
of dimension 3 having nonlinear automorphisms, since it does not contribute any insight
to the discussion in the main paper.

\bibliographystyle{plain}
\bibliography{main}

\begin{thebibliography}{10}

\bibitem{Tofu}
Yuichiro Ajima, Shinji Sumimoto, and Toshiyuki Shimizu.
\newblock Tofu: A {6D} mesh/torus interconnect for exascale computers.
\newblock {\em Computer}, 42:36--40, 2009.

\bibitem{Krishna}
S.B. Akers and B.~Krishnamurthy.
\newblock A group-theoretic model for symmetric interconnection networks.
\newblock {\em IEEE Transactions on Computers}, 38:555--566, 1989.

\bibitem{IlliacIV}
G.H. Barnes, R.M. Brown, M.~Kato, D.J. Kuck, D.L. Slotnick, and R.A. Stokes.
\newblock The {I}lliac {IV} computer.
\newblock {\em IEEE Transactions on Computers}, C-17(8):746--757, aug. 1968.

\bibitem{Beivide}
Ram\'{o}n Beivide, Enrique Herrada, Jos\'{e}~L. Balc\'{a}zar, and Agustin
  Arruabarrena.
\newblock Optimal distance networks of low degree for parallel computers.
\newblock {\em IEEE Trans. Comput.}, 40(10):1109--1124, 1991.

\bibitem{jaguar}
Buddy Bland.
\newblock Jaguar: Powering and cooling the beast.
\newblock
  {\url{http://www.cse.ohio-state.edu/~panda/875/class_slides/cray-jaguar.pdf}},
  2009.

\bibitem{Cai}
Jin{-}yi Cai, George Havas, Bernard Mans, Ajay Nerurkar, Jean-Pierre Seifert,
  and Igor Shparlinski.
\newblock On routing in circulant graphs.
\newblock In {\em COCOON}, pages 360--369, 1999.

\bibitem{Camara}
Jose~M. Camara, Miquel Moreto, Enrique Vallejo, Ramon Beivide, Jose
  Miguel-Alonso, Carmen Mart\'{i}nez, and Javier Navaridas.
\newblock Twisted torus topologies for enhanced interconnection networks.
\newblock {\em IEEE Transactions on Parallel and Distributed Systems},
  21:1765--1778, 2010.

\bibitem{IWONT10}
C.~Camarero, C.~Mart\'{\i}nez, and R.~Beivide.
\newblock Symmetric {L}-graphs.
\newblock In {\em 2010 International Workshop on Optimal Network Topologies},
  2010.

\bibitem{CamareroTC}
C.~Camarero, C.~Mart\'{i}nez, and R.~Beivide.
\newblock L-networks: A topological model for regular two-dimensional
  interconnection networks.
\newblock {\em Computers, IEEE Transactions on}, PP(99):1, 2012.

\bibitem{HPC}
Crist\'{o}bal Camarero, Enrique Vallejo, Carmen Mart\'{i}nez, Miquel Moreto,
  and Ram\'{o}n Beivide.
\newblock Task mapping in rectangular twisted tori.
\newblock In {\em 21st High Performance Computing Symposia (HPC'13), Part of
  the SCS Spring Simulation Multiconference (SpringSim'13)}, 2013.
\newblock to apear.

\bibitem{Chen}
Dong Chen, Noel Eisley, Philip Heidelberger, Sameer Kumar, Amith Mamidala,
  Fabrizio Petrini, Robert Senger, Yutaka Sugawara, Robert Walkup, Burkhard
  Steinmacher-Burow, Anamitra Choudhury, Yogish Sabharwal, Swati Singhal, and
  Jeffrey~J. Parker.
\newblock Looking under the hood of the {IBM} {Blue Gene/Q} network.
\newblock In {\em Proceedings of the International Conference on High
  Performance Computing, Networking, Storage and Analysis}, SC '12, pages
  69:1--69:12, Los Alamitos, CA, USA, 2012. IEEE Computer Society Press.

\bibitem{BGQnetwork}
Dong Chen, Noel~A. Eisley, Philip Heidelberger, Robert~M. Senger, Yutaka
  Sugawara, Sameer Kumar, Valentina Salapura, David~L. Satterfield, Burkhard
  Steinmacher-Burow, and Jeffrey~J. Parker.
\newblock The {IBM} {Blue Gene/Q} interconnection network and message unit.
\newblock In {\em Proceedings of 2011 International Conference for High
  Performance Computing, Networking, Storage and Analysis}, SC '11, pages
  26:1--26:10, New York, NY, USA, 2011. ACM.

\bibitem{Coteus}
P.~Coteus, H.~R. Bickford, T.~M. Cipolla, P.~G. Crumley, A.~Gara, S.~A. Hall,
  G.~V. Kopcsay, A.~P. Lanzetta, L.~S. Mok, R.~Rand, R.~Swetz, T.~Takken, P.~La
  Rocca, C.~Marroquin, P.~R. Germann, and M.~J. Jeanson.
\newblock Packaging the {B}lue {G}ene/{L} supercomputer.
\newblock {\em IBM Journal of Research and Development}, 49(2.3):213 --248,
  march 2005.

\bibitem{Cray}
Cray {XE6} brochure.
\newblock {\url{http://www.cray.com/Products/XE/Technology.aspx}}.

\bibitem{Cascade}
Greg Faanes, Abdulla Bataineh, Duncan Roweth, Tom Court, Edwin Froese, Bob
  Alverson, Tim Johnson, Joe Kopnick, Mike Higgins, and James Reinhard.
\newblock Cray cascade: a scalable {HPC} system based on a dragonfly network.
\newblock In {\em Proceedings of the International Conference on High
  Performance Computing, Networking, Storage and Analysis}, SC '12, pages
  103:1--103:9, Los Alamitos, CA, USA, 2012. IEEE Computer Society Press.

\bibitem{Fiol}
M.A. Fiol.
\newblock On congruence in {$\mathbb Z^n$} and the dimension of a
  multidimensional circulant.
\newblock {\em Discrete Math}, 141:1--3, 1995.

\bibitem{Flahive}
Mary Flahive and Bella Bose.
\newblock The topology of {G}aussian and {E}isenstein-{J}acobi interconnection
  networks.
\newblock {\em IEEE Trans. Parallel Distrib. Syst.}, 21(8):1132--1142, 2010.

\bibitem{Gomez-Perez}
Domingo G{\'o}mez, Jaime Gutierrez, {\'A}lvar Ibeas, Carmen Mart\'{\i}nez, and
  Ram{\'o}n Beivide.
\newblock On finding a shortest path in circulant graphs with two jumps.
\newblock In {\em COCOON}, pages 777--786, 2005.

\bibitem{Janssen}
T.~Janssen.
\newblock {\em Crystallographic Groups}.
\newblock American Elsevier, 1973.

\bibitem{Martin}
A.~J. Martin.
\newblock The torus: An exercise in constructing a processing surface.
\newblock {\em Proceedings of the VLSI Conference}, 1981.

\bibitem{IEEETITQuat}
C.~Mart\'{i}nez, R.~Beivide, and E.M. Gabidulin.
\newblock Perfect codes from {C}ayley graphs over {L}ipschitz integers.
\newblock {\em Information Theory, IEEE Transactions on}, 55(8):3552 --3562,
  aug. 2009.

\bibitem{IEEE_TC}
Carmen Mart\'{i}nez, Ramon Beivide, Esteban Stafford, Miquel Moreto, and
  Ernst~M. Gabidulin.
\newblock Modeling toroidal networks with the {G}aussian integers.
\newblock {\em IEEE Transactions on Computers}, 57:1046--1056, 2008.

\bibitem{Navaridas}
Javier Navaridas, Jose Miguel-Alonso, Jose~A. Pascual, and Francisco~J.
  Ridruejo.
\newblock Simulating and evaluating interconnection networks with insee.
\newblock {\em Simulation Modelling Practice and Theory}, 19(1):494 -- 515,
  2011.
\newblock Modeling and Performance Analysis of Networking and Collaborative
  Systems.

\bibitem{Newman}
Morris Newman.
\newblock {\em Integral matrices}.
\newblock Academic Press, New York,, 1972.

\bibitem{XeonRing}
Cheolmin Park, R.~Badeau, L.~Biro, J.~Chang, T.~Singh, J.~Vash, Bo~Wang, and
  T.~Wang.
\newblock A 1.2 {TB/s} on-chip ring interconnect for 45nm 8-core enterprise
  {Xeon\circledR} processor.
\newblock In {\em Solid-State Circuits Conference Digest of Technical Papers
  (ISSCC), 2010 IEEE International}, pages 180 --181, feb. 2010.

\bibitem{Robic}
Borut Robic.
\newblock Optimal routing in 2-jump circulant networks.
\newblock Technical report, University of Cambridge Computer Laboratory, TR397,
  1996.

\bibitem{Sequin}
Carlo~H. Sequin.
\newblock Doubly twisted torus networks for {VLSI} processor arrays.
\newblock In {\em ISCA '81: Proceedings of the 8th annual symposium on Computer
  Architecture}, pages 471--480, Los Alamitos, CA, USA, 1981. IEEE Computer
  Society Press.

\bibitem{Zerovnik}
Janez Zerovnik.
\newblock Perfect codes in direct products of cycles--a complete
  characterization.
\newblock {\em Advances in Applied Mathematics}, 41(2):197 -- 205, 2008.

\end{thebibliography}

%
%

\end{document}